\documentclass[10pt,twocolumn]{IEEEtran}
\usepackage{url}
\usepackage{array}
\usepackage{makecell}
\usepackage[usenames,dvipsnames]{pstricks}
\usepackage{epsfig}
\usepackage{pst-grad} 
\usepackage{pst-plot} 
\usepackage{amsmath,epsfig}
\usepackage{amssymb}
\usepackage{mathtools}
    \DeclarePairedDelimiter\ceil{\lceil}{\rceil}
    
\usepackage{enumitem}
\usepackage{psfrag}
\usepackage{color}
\usepackage{pstricks}
\usepackage{cite}
\usepackage{algorithm}
\usepackage{epsfig}
\usepackage{amsthm}
\usepackage{amssymb}
\usepackage{psfrag}
\usepackage{pstricks}
\usepackage{float}
\usepackage{stfloats}
\usepackage{wrapfig}
\usepackage{graphicx}
\usepackage{amsmath}
\usepackage{amsfonts}
\usepackage{amssymb}
\usepackage{algorithm}
\usepackage{cuted}
\usepackage{blindtext}
\usepackage{pstricks,pst-plot}
\usepackage{pst-bezier}
\usepackage{pst-math}
\usepackage{algorithm,algcompatible,amsmath}
\usepackage{hyperref}
\usepackage{float}
\usepackage{cite}
\usepackage{amsmath,amssymb,amsfonts,dsfont}
\usepackage{scalerel,stackengine}
\stackMath
\newcommand\reallywidehat[1]{%
\savestack{\tmpbox}{\stretchto{%
  \scaleto{%
    \scalerel*[\widthof{\ensuremath{#1}}]{\kern-.6pt\bigwedge\kern-.6pt}%
    {\rule[-\textheight/2]{1ex}{\textheight}}
  }{\textheight}%
}{0.5ex}}%
\stackon[1pt]{#1}{\tmpbox}%
}
\usepackage{algorithm}
\usepackage[english]{babel}
\usepackage[utf8]{inputenc}
\usepackage[noend]{algpseudocode}
\usepackage{graphicx}
\usepackage{textcomp}
\usepackage{setspace}
\usepackage [english]{babel}
\usepackage [autostyle, english = american]{csquotes}
\usepackage[dvipsnames]{xcolor}
\colorlet{LightRubineRed}{RubineRed!70!}
\colorlet{Mycolor1}{green!10!orange!90!}
\definecolor{Mycolor2}{HTML}{0000FF} 
\definecolor{Mycolor3}{HTML}{000000} 
\usepackage{tikz}
\def\checkmark{\tikz\fill[scale=0.4](0,.35) -- (.25,0) -- (1,.7) -- (.25,.15) -- cycle;}

\algnewcommand\REQUIRED{\item[\textbf{Required:}]}%
\algnewcommand\INPUT{\item[\textbf{Input:}]}%
\algnewcommand\OUTPUT{\item[\textbf{Output:}]}%

\def\bg{{\bf g}}

\def\bc{{\bf c}}

\def\bo{{\bf o}}

\def\bm{{\bf m}}

\def\bs{{\bf s}}

\def\bx{{\bf x}}

\def\rs{{\mathsf{s}}}

\def\ro{{\mathsf{o}}}

\def\rx{{\mathsf{x}}}
\def\rc{{\mathsf{c}}}
\def\rm{{\mathsf{m}}}
\def\rtr{{\mathsf{tr}}}

        \newtheorem{theorem}{Theorem}

        \newtheorem{condition}[theorem]{Condition}

        \newtheorem{definition}[theorem]{Definition}

        \newtheorem{lemma}[theorem]{Lemma}

\title{\LARGE Task-Oriented Data Compression for Multi-Agent Communications Over Bit-Budgeted Channels \\ 
\thanks{The authors are with the Centre for Security Reliability and Trust, University of Luxembourg, Luxembourg. Emails: \{arsham.mostaani, thang.vu, symeon.chatzinotas, bjorn.ottersten\}@uni.lu}
\thanks{This work is supported by the ERC AGNOSTIC project, ref. H2020/ERC2020POC/957570/DREAM.}}
\author{\IEEEauthorblockN{Arsham Mostaani, Thang X. Vu, Symeon Chatzinotas, and Bj\"orn Ottersten}
}
\begin{document}
\maketitle

\vspace{-10mm}
\begin{abstract}
\textcolor{Mycolor3}{Various applications for inter-machine communications are on the rise. Whether it is for autonomous driving vehicles or the internet of everything, machines are more connected than ever to improve their performance in fulfilling a given task. While in traditional communications the goal has often been to reconstruct the underlying message, under the emerging task-oriented paradigm, the goal of communication is to enable the receiving end to make more informed decisions or more precise estimates/computations. Motivated by these recent developments, in this paper, we perform an indirect design of the communications in a multi-agent system (MAS) in which agents cooperate to maximize the averaged sum of discounted one-stage rewards of a collaborative task. Due to the bit-budgeted communications between the agents, each agent should efficiently represent its local observation and communicate an abstracted version of the observations to improve the collaborative task performance. We first show that this problem can be approximated as a form of data-quantization problem which we call task-oriented data compression (TODC). We then introduce the state-aggregation for information compression algorithm (SAIC) to solve the formulated TODC problem. It is shown that SAIC is able to achieve near-optimal performance in terms of the achieved sum of discounted rewards. The proposed algorithm is applied to a geometric consensus problem and its performance is compared with several benchmarks. Numerical experiments confirm the promise of this indirect design approach for task-oriented multi-agent communications. }
\end{abstract}


\vspace{-3mm}
\section{Introduction}

The design of traditional communication systems has often been carried out according to task-agnostic principles. Information and coding theories drive the major analytical and design techniques, where the former sets the upper bounds on the system capacity, and the latter focuses on techniques for approaching the bounds with infinitesimal error probabilities. Accordingly, digital communications have made astonishing strides in terms of performance, enabling robust information transmission even under adverse channel conditions. However, in the era of cyber-physical systems, the effectiveness of communications is not solely dictated by the traditional performance indicators (e.g., bit rate, latency, jitter, fairness etc.), but most importantly by the efficient completion of the task in hand, e.g., remotely controlling a robot, automating a production line or collaboratively sensing/communicating through a drone swarm.

 Machine to machine communications occur since the received signals can help the receiving end to make more informed decisions or more precise estimates/computations. In this context, the reliability of the communications is not essential beyond serving the specific needs of the control/estimation/computational task that the receiving end machine is trying to accomplish. This calls for a fresh look into the design of communication systems that have been engineered with reliability as one of their ultimate goals. The emerging literature on semantic communications as well as goal/task-oriented communications is trying to take the first steps towards the above-mentioned goal, i.e., incorporating the semantics as well as the goal/usefulness of the message exchange into the design of communication systems \cite{mostaani2021task,gunduz2022beyond, strinati20216g}.
By jointly analyzing the features of the collaborative task and the constraints on the underlying communication infrastructure, the communication strategies can be adapted or tailored such that they will be specifically effective for the task. 

\textcolor{Mycolor3}{ This paper attempts to take the first steps towards designing an \emph{indirect} task-effective data compression theory. While the data compression algorithm proposed by this paper is designed in an \emph{indirect}\footnote{ \textcolor{Mycolor3}{ By using the word \emph{indirect} here we are not referring to the concept of indirect access to the source of information \cite{1056251} - this usage of the word falls in the nomenclature of source coding and information theory. In fact, we are referring to the concept being introduced by the control theory nomenclature in which an indirect design is generic enough to be used for an unmodelled system dynamics and not a certain dynamic \cite{ioannou1988theory}. Thus the schemes - such as SAIC - which enjoy from an indirect design can be applied to all/a wider range of tasks. In contrast to indirect schemes, "the direct schemes aim at guaranteeing or improving the performance of the cyber-physical system at a particular task by designing a
task-tailored communication strategy" \cite{mostaani2021task}. }} fashion i.e., not for a specific task, we demonstrate its applicability in a specific task: a geometric consensus problem under finite observability \cite{barel2017come}. As attested by \cite{xie2022task}, "a unified framework to support various tasks is still missing in multi-user semantic communications.". Unlike earlier task-oriented quantization techniques that tailor a quantization scheme to certain application \cite{shlezinger2020task}, this work proposes an \emph{indirect} design for its task-oriented quantization scheme - SAIC. The \emph{indirect} design is carried out in a fashion that the it never benefits from any explicit domain knowledge about any specific task e.g., geometric consensus problems. Accordingly, the \emph{indirect} design of the algorithms allows them to be applied beyond the geometric consensus problems and to a much wider range of tasks.} The framework can be applied where a major communication bottleneck is in place between multiple cooperative decision makers. This bottleneck can occur due to a multitude of reasons (i) the energy lifetime of the communicating agents e.g., in the case of UAV/LEO satellite communications, that forces agents to communicate with low-energy high-range communication protocols \cite{palattella2018enabling, chaari2019heterogeneous} (ii) the limitations imposed by the environment on the communication channel e.g., in space/underwater missions or (iii) limited communication resources of the network through which agents communicate. For more on the applications of TODC see \cite{mostaani2021task, azari2022evolution}.

\subsection{Task-Oriented Data Compression}
In particular, we consider a cooperative scenario where our goal is to optimize the expected return of a multi-agent system that is run on top of an underlying Markov decision process. The system's return is an unknown function of joint observations and control actions of all agents. The system's expected return can be controlled or optimized by selecting the proper joint controls actions at all agents. The partial observability of each agent together with their limitation to merely select local actions necessitates the presence of inter-agent communications to improve the coordination across the multi-agent system. We assume a full mesh communication network between all agents and that all the communication channels in the network are bit-budgeted but error-free. \textcolor{Mycolor3}{ That is, the communication channels are all error-free fixed-rate bit pipes \cite{nair2003exponential} and not variable rate bit-pipes \cite{nair2004stabilizability} - the fixed rate of communications is constant across all inter-agent communication channels.} Under these circumstances, rate-limited communication channels between agents drive the need for task-oriented data compression i.e., the usefulness of each message exchange should be incorporated into the design of the data compression strategy. The communicated messages between agents are useful only when they positively affect the decision-making of the receiving agents towards improving the system's expected return.

\textcolor{Mycolor3}{
The problem we address would be a classic multi-agent Markov decision process (MAMDP) \cite{lauer2000distributedQ} if, each agent's communication message could include all the information inside the agent's observations.  We assume, however, that the communication message of each agent is sent over a bit-budgeted communication channel i.e., per each channel use each agent will be able to reliably communicate a bit sequence with a length less than the entropy rate of the observation process. With this information constraint in place, it becomes imperative to carry out the communications at each agent such that they lead to the optimal expected return performance of the MAS. Each agent has to jointly select its control and quantized message at each time step with the aim of optimizing the expected return.}

\textcolor{Mycolor3}{
Due to the bit-budgeted communications between the agents, it is necessary for agents to compactly represent their observations in communication messages. As we ultimately measure the performance of the MAS in terms of the expected return, the loss of information caused by the compact representation of the agents' observations needs to be managed in such a way that it minimally affects the obtained return \cite{kostina2019rate, tung2021effective}. As such, in this form of compression scheme which we call task-oriented data compression, \textit{ the goal of abstraction is different from conventional compression schemes} whose ultimate aim is to reduce the distortion between the original signal and the decoded/reconstructed signal \cite{arimoto1972algorithm} - see \cite{shlezinger2020task,shlezinger2021deep}, where a similar task-based notion is introduced and a comparison of it with our work in Table \ref{table: related-works}. }

\vspace{-3mm}
\subsection{Literature Review}
\vspace{-2mm}
\textcolor{Mycolor3}{As we study the joint communication and control design of a MAS, the topic of this paper falls under the general category of multi-agent communications \cite{pynadath2002communicative}.} \textcolor{Mycolor3}{ In contrast to many other cooperative multi-agent systems \cite{lee2020optimization}, the  full state and action information are not available here to each agent.  Accordingly, agents are required to carry out communication to overcome these barriers \cite{pynadath2002communicative}.} Earlier works used to address the coordination of multiple agents through a noise-free communication channel, where the agents follow an engineered communication strategy \cite{ZhangCoordinating,fischer2004hierarchical,kasai2008learning,wu2011online,Amiri2018CEASE}. Later the impact of stochastic delays in multi-agent communication was considered on the multi-agent coordination \cite{wu2011online}, while \cite{Amiri2018CEASE} considers event-triggered local communications. Deep reinforcement learning with communication of the gradients of the agents' objective function was proposed in \cite{FoersterLearning} to learn the communication among multiple agents. In contrast to the above-mentioned works, the presence of noise in the inter-agent communication channel was first studied by \cite{mostaani2019Learning} where exact reinforcement learning was used to design the inter-agent communications. \textcolor{Mycolor3}{ Later, the authors of \cite{tung2021effective} proposed a deep reinforcement learning approach to address a similar problem. Papers \cite{mostaani2019Learning, mostaani2020state, tung2021effective, shlezinger2021deep, shlezinger2020task} and \cite{kim2019schedule} have contributed to the rapidly emerging literature on task-oriented communications \cite{mostaani2021task}.} Noteworthy are also some novel metrics that are introduced in \cite{lowe2019pitfalls} to measure the \textit{positive signaling} and \textit{positive listening} amongst agents which learn how to communicate \cite{mostaani2019Learning,kim2019schedule,FoersterLearning}.

The current work can also be seen as designing a state aggregation algorithm. In this paper, state aggregation enables each agent to compactly represent its observations through communication messages while maintaining their performance in the collaborative task. Classical state aggregation algorithms, however, have been used to reduce the complexity of the dynamic programming problems over MDPs \cite{Bertsekas1989AdaptiveAg,bertsekas2018feature, abel2016near, rubino1989weak} as well as Partially Observable MDPs \cite{bertsekas2019biased}. One similar work is \cite{zou2018decision}, which studies a task-based quantization problem. In contrast to our work, the assumption there is that the parameter to be quantized is only measurable and cannot be controlled. In our problem, agents' observations stem from a generative process with memory, an MDP. Similarly, in \cite{mao2019learning}, the authors have introduced a gated mechanism so that reinforcement learning-aided agents reduce the rate of their communication by removing messages which are not beneficial for the team. However, their proposed approach mostly relies on numerical experiments. In contrast, this paper relies on analytical studies to design a multi-agent communication policy which efficiently coordinates agents over a bit-budgeted channel - \textcolor{Mycolor3}{ the benefits of our analytical approach are briefly explained in the contributions section \ref{subsect: contributions}. }
State aggregation algorithms are often developed for single-agent scenarios and are used to reduce the complexity of MDPs. \textcolor{Mycolor3}{To the best of our knowledge, we are the first to design a TODC algorithm using state-aggregation schemes. In particular, we use state-aggregation to design a data compression scheme to compactly represent the observation process of each agent in a multi-agent system.}

Conventionally, the communication system design is disjoint from the distributed decision-making design \cite{ZhangCoordinating,fischer2004hierarchical,kasai2008learning,wu2011online,FoersterLearning,sukhbaatar2016learning}. The current work can also be interpreted as a demonstration of the potential of the joint design of the data compression/quantization and control policies. Determining the existence of a quantizer operating at a certain bit-budget to achieve a given figure of expected return is known to be an "intriguing open problem" \cite{kostina2019rate} - even for single agent scenarios. Here we set a non-closed form upper bound on the expected-return performance of the multi-agent system given a quantization data rate/ the finite size of the discrete alphabet of the quantizer. We show how this joint quantization and control design problem is connected to minimizing an absolute error distortion measure via Theorem 1. A similar interpretation of the TODC problem can also be seen in \cite{stavrou2022rate}. While relevant, their setup is different from our work as they consider two distortion criteria for the rate-distortion problem.

\textcolor{Mycolor3}{
We will show in section \ref{problem statement subsection}, that, in fact, the decentralized problem we target can be translated as the joint constrained design of the control policies as well as the observation function of a Dec-POMDP to maximize the expected return. While in classic Dec-POMDP problems the observation function is considered to be a fixed function \cite{oliehoek2008optimal}, by a constrained design of the observation function, our problem setting offers more flexibility in designing a multi-agent system. The design of the observation function helps to filter the non-useful observation information of each agent while meeting the problem's constraint i.e., the communication bit-budget. The mathematical framework being used here is neither a classic MDP as we have the issue of partial observability, nor is a partially observable MDP (POMDP) \cite{monahan1982state} as the action vector is not jointly selected at a single entity.} \textcolor{Mycolor3}{ Our problem setting is differentiated  from Dec-POMDPs due to the fact that in Dec-POMDPs the partial observability is accepted as is, where as in our problem setting we design the lens through which the agents acquire a partial observation/perception of the environment. }

\textcolor{Mycolor3}{
Nevertheless, a similar class of problems - often referred to as task-oriented, goal-oriented or efficient communication approaches, has recently received significant attention from the communication society, see e.g., the extensive surveys on similar problems in \cite{mostaani2021task,gunduz2022beyond, strinati20216g}. Table \ref{table: related-works} positions the current work against some of the recent research that is closely related. To date, there is no work in the literature that we are aware of, which provides an analytical approach to the design of task-based communications for the coordination of multiple cooperative agents.}


\vspace{-3mm}
\subsection{Contributions} \label{subsect: contributions}
\vspace{-1mm}
The contributions of this paper are as follows:

    \textcolor{Mycolor3}{Firstly, we develop a general cooperative multi-agent framework in which agents interact over an underlying MDP environment. Unlike the existing works which assume perfect communication links \cite{zilber2001communication,sukhbaatar2016learning,FoersterLearning,kim2019schedule}, we assume the practical bit-budgeted communications between the agents. We formulate a multi-agent cooperative problem where agents interact over an underlying MDP and can communicate over a bit-budgeted channel. Our goal is to derive the optimal control and communication strategies to maximize the expected return. We will show in section \ref{problem statement subsection}, that an underlying difference in our setting from the Dec-POMDP is that here we carry out a constrained design of each agent's perception function - which is also referred to as the observation function in the literature of the Dec-POMDP \cite{oliehoek2007dec}. The constraints of this design are dictated by the bit-budget of the inter-agent communication channels.}  
    
    \textcolor{Mycolor3}{Secondly, Theorem 1, in section \ref{sec: SAIC}, derives the interconnection between the joint control and communication/quantization problem and a generalized version of the data quantization problem: TODC  problem. In fact, the TODC problem distils all the relevant features of the control task and takes them into account in a novel non-conventional communication design problem. This is the underlying reason behind the effectiveness of the designed communications and is one the contributions in this work differentiating it from existing works in \cite{shlezinger2021deep,shlezinger2020task,larrousse2018coordination,kostina2019rate,FoersterLearning, lowe2019pitfalls,tung2021effective, mostaani2019Learning}.}
    \textcolor{Mycolor3}{Our analytical studies show that how the value function - the function that estimates the expected return of the system given the current observation - can be considered as a proper indirect measure of the usefulness of the data to be compressed. Thus, Theorem 1, shows how the usefulness of the (observation) data can be incorporated into the design of the TODC policy.}
    

    \textcolor{Mycolor3}{Thirdly, we  propose a novel algorithm - SAIC - as a multi-agent state-aggregation algorithm which designs indirect task-effective communication strategies via solving (an approximated version of) the TODC problem. As a result, the performance of SAIC in terms of the system's expected return is on par with the jointly optimal strategies. To the best of our knowledge, this is the first use of state-aggregation algorithms for data-compression applications (in multi-agent systems) according to which our work differs from the classic state-aggregation literature \cite{Bertsekas1989AdaptiveAg,bertsekas2018feature, abel2016near, rubino1989weak} as well as the recent advancements in multi-agent communication literature \cite{FoersterLearning, lowe2019pitfalls}.}
    
    \textcolor{Mycolor3}{Moreover, we extend the existing results in the single-agent state-aggregation literature \cite{abel2016near} on the gap between the optimal control and the state-aggregated control schemes, where the former has access to the true state of the environment and the latter has access to an aggregated state of the environment - to reduce the computational complexity. We quantify the same gap for a multi-agent system - Theorem \ref{theorem: error bound}. In our work, however, the gap is due to the bit-budget that is introduced on the inter-agent communication channels, whereas in classic state-aggregation literature the gap was a consequence of the constraints on the computational complexity. In addition to that, our theoretical results show that if our proposed method, SAIC, is applied the expected return of the multi-agent communication system - with the bit-budget in place - can stay in close proximity to the optimal expected return that is obtained under jointly optimal strategies.}
    
\textcolor{Mycolor3}{  
\begin{table*}
\caption{\textcolor{Mycolor3}{Comparison between our work and the related prior art}}
\centering
 \begin{tabular}{||c c c c c c c ||} 
 \hline
 Paper & \begin{tabular}[c]{@{}c@{}}Information\\  source\\ with memory\end{tabular} & \begin{tabular}[c]{@{}c@{}}Joint coms \\ and control\end{tabular} & Distributed & \begin{tabular}[c]{@{}c@{}}Source/Channel\\  coding\\ Quantization\end{tabular} & \begin{tabular}[c]{@{}c@{}}Implicit/Explicit\\ coms\end{tabular} & \begin{tabular}[c]{@{}c@{}}Analytical/\\ Data-driven \end{tabular}  \\ [0.5ex]  \hline\hline
 \small{\cite{shlezinger2021deep,shlezinger2020task}} & \small{$\times$} &  \small{$\times$} & \small{$\times$} & Quantization & N/A & Data-driven \\ 
 \hline
  \small{\cite{larrousse2018coordination}} & \small{$\times$} &  \small{$\checkmark$} & \small{$\checkmark$} & N/A & Implicit & Analytical \\ 
 \hline
   \small{\cite{kostina2019rate}} & \small{$\checkmark$ (Linear)} &  \small{$\checkmark$} & \small{$\times$} & Quantization & Explicit & Analytical \\ 
 \hline
    \small{\cite{FoersterLearning, lowe2019pitfalls}} & \small{N/A} &  \small{$\checkmark$} & \small{$\checkmark$} & N/A & Explicit & Data-driven \\ 
 \hline
 \small{\cite{tung2021effective, mostaani2019Learning} } & \small{\checkmark (Markov)} &  \small{$\checkmark$} & \small{$\checkmark$} & \small{Channel Coding} & Explicit & Data-driven \\ 
 \hline
  \small{Our work} & \small{\checkmark (Markov)} &  \small{$\checkmark$} & \small{$\checkmark$} & \small{Quantization} & Explicit & Analytical \\ 
 \hline
\end{tabular}
\label{table: related-works}
\vspace{-3mm}
\end{table*}}
    \textcolor{Mycolor3}{Last but not least, numerical experiments are carried out on a geometric consensus problem to compare the performance of SAIC with several other benchmark schemes in terms of the optimality of the expected return, for a multi-agent scenario \footnote{ Due to the complexity related issues explained in section \ref{Numerical results - Section} \& \ref{sec: conclusion}, the numerical results are limited to two-agent and three-agent scenarios.}.
    It is shown that when communication bit-budgets are in place, SAIC is of significant advantage over the benchmarks. In particular, we observe a very tight gap between the performance of SAIC and the optimal control strategy where only the latter runs over perfect communication channels and the former runs over bit-budgeted channels.}

\subsection{Organization} Section II describes the system model for a cooperative multi-agent task with rate-constrained inter-agent communications. Section III Proposes a scheme for the joint design of communication and control policies that takes the value of information into account to perform data compression. \textcolor{Mycolor3}{We also provide analytical results on how distant the result of this algorithm can be from the optimal centralized solution. The numerical results and discussions are provided in section IV. Finally, section V concludes the paper.}

\subsection{Notation}
For the reader's convenience, a summary of the notation that we follow in this paper is given in Table \ref{table-notation}. Bold font is used for matrices or scalars which are random and their realizations follows simple font.

\vspace{-2mm}
 \section{System Model} \label{System model - Section}
 \vspace{-2mm}


   In the multi-agent system, comprised of $n$ agents, at any time step $t$ each agent $i \in \mathcal{N}$ makes a local observation $\ro_i(t) \in \Omega$ on environment while the true state of the environment
  \begin{align} \label{eq: state vs observation}
      \rs(t) = \langle \ro_1(t),..., \ro_n(t) \rangle
  \end{align}
  is a member of $ \mathcal{S} = \Omega^n$. The alphabets $\Omega $ and $\mathcal{S}$ define observation space and state space, respectively. The particular observation structure of agents' observations, is referred to as collective observations in the literature \cite{pynadath2002communicative}. Under collective observability, individual observation of an agent provides it with partial information about the current state of the environment, however, having knowledge of the collective observations acquired by all of the agents is sufficient to realize the true state of environment - eq. (\ref{eq: state vs observation}). The columns of the state vector are orthogonal to each other. Note that even in the case of collective observability, for agent $i$ to be able to observe the true state of environment at all times, it needs to have access to the observations of the other agents $j \in \mathcal{N} - \{i\} \triangleq \mathcal{N}_{-i} $ through communications at all times.

  The true state of the environment $\rs(t)$ is controlled by the joint actions $\rm(t) = \langle \rm_1(t), ..., \rm_n(t) \rangle \in \mathcal{M}^n$ of the agents, where each agent $i$ can only choose its local action $\rm_i(t) \in \mathcal{M}$. The environment runs on discrete time steps $t = 1, 2, ..., M$, where at each time step, each agent $i$ selects its domain level action $\rm_i(t)$ upon having an observation $\ro_i(t)$ of the environment. \textcolor{Mycolor3}{ Dynamics of the environment are governed by a conditional probability mass function (CMF) 
\begin{small}
 \begin{align} \label{transition probability - equation}
& T\big(\rs(t+1) | \rs(t), \rm(t) \big) = p\big(\rs(t+1) |  \rs(t), \rm(t)\big)
\end{align}
 \end{small}
 which is unknown to the agents. $T(\cdot): \Omega^{2n} \times \mathcal{M}^n \rightarrow [0,1]$ determines the future state of the environment $\rs(t+1)$ given its current state $\rs(t)$ and the joint actions $\rm(t)$.} We recall that each agent $i$'s domain level action $\rm_i(t)$ can, for instance, be in the form of a movement or acceleration in a particular direction or any other type of action depending on the domain of the cooperative task.

A deterministic reward function $r(\cdot): \Omega^n \times \mathcal{M}^n \rightarrow \mathbb{R}$ indicates the reward of all agents at time step $t$, where the arguments of the reward function are the joint observations $\rs(t)$ and the domain-level joint actions $\rm(t)$ of all agents. We assume that the underlying environment over which agents interact can be defined in terms of an MDP \footnote{As defined in the literature [10], the underlying MDP’ is the horizon-$T'$ MDP defined by a hypothetical single agent that takes joint actions
$\bm(t) \in \mathcal{M}^n $ and observes the nominal state $\bs(t) \triangleq \langle \bo_1(t),\dots, \bo_n(t)   \rangle $ that has the same transition model $T(\cdot)$ and reward
model $R(\cdot)$ as the environment experienced by our multi-agent system.} determined by the tuple $\big{\{} \Omega^n, \mathcal{M}^n, r(\cdot), \gamma, T(\cdot)  \big{\}}$, where $\Omega$ and $\mathcal{M}$ are discrete alphabets, $r(\cdot)$ is a function, $T(\cdot)$ is defined in (\ref{transition probability - equation}) and the scalar $\gamma \in [0,1]$ is the discount factor. The focus of this paper is on scenarios in which the agents are unaware of the state transition probability function $T(\cdot)$ and of the closed form of the function $r(\cdot)$. However we assume that, further to the literature of reinforcement learning \cite{Suttonintroduction}, a realization of the function $r\big(\rs(t),\rm(t)\big)$ will be accessible for all agents at some time steps. Since the tuple $\big{\{} \Omega^n, \mathcal{M}^n, r(\cdot), \gamma, T(\cdot)  \big{\}}$ is an MDP and the state process $\bs(t)$ is jointly observable by agents, the system model of this cooperative multi-agent setting, under perfect communications, is also referred to as a  multi-agent MDP (MAMDP or MMDP) in the literature of multi-agent decision making \cite{rizk2018decision, lauer2000distributedQ, boutilier1999multiagent}.

In what follows two problems regarding the above-mentioned setup is detailed i.e., centralized and decentralized control problems. The main intention of this paper is to address decentralized control which also incorporates inter-agent communications for a system of multiple agents. \textcolor{Mycolor3}{ The centralized control problem, however, is also formalized in subsection \ref{subsec: centralized control} as the optimal expected return obtained for the centralized problem can serve as a lower-bound/(upper-bound) for the decentralized scheme.} Moreover, the simpler nature and mathematical notations used for the centralized problem, allow the reader to have a smoother transition to the decentralized problem which is of a more complex nature.

\begin{table}[b]
\caption{Table of notations}
\centering
 \begin{tabular}{||c c ||} 
 \hline
 Symbol & Meaning \\ [0.5ex] 
 \hline\hline
 \small{$\bx(t)$} & \small{A generic random variable generated at time $t$}  \\ 
 \hline
 \small{$\rx(t)$} & \small{Realization of $\bx(t)$}  \\
 \hline
 \small{$\mathcal{X}$} & \small{Alphabet of \bx(t)}  \\
 \hline
 \small{$|\mathcal{X}|$} & \small{Cardinality of $\mathcal{X}$}  \\
 \hline
 \small{$p_{\bx}\big(\rx(t)\big)$} & \small{Shorthand for $\mathrm{Pr}\big(\bx(t) = \rx(t) \big)$}  \\  
 \hline
 \small{$H\big(\bx(t)\big)$} & \small{Information entropy of $\bx(t) $ (bits)}  \\  
 \hline
  \small{$\mathcal{X}_{-\bx}$} & \small{ $\mathcal{X} - \{\bx\}$}  \\ [1ex]
 \hline
   \small{$\mathbb{E}_{p(\rx)}\{\bx\}$} & \small{\makecell{Expectation of the random variable $X$ over the \\ probability distribution $p(\rx)$}}  \\ [1ex]
 \hline
    \small{$\delta(\cdot)$} & \small{Dirac delta function}  \\ [1ex]
 \hline
   \small{$\rtr(t) $} & \small{Realization of the system's trajectory at time $t$}  \\ [1ex]
 \hline
\end{tabular}
\label{table-notation}
\end{table}
\vspace{-3mm}
\subsection{Centralized Control} \label{subsec: centralized control}
\vspace{-1mm}
We consider a scenario in which a central controller has instant access to the observations $\ro_1(t), ..., \ro_n(t)$ of both agents through a free (with no cost on the objective function) and reliable communication channel.
From the central controller's point of view, the environment is the same as the underlying MDP that governs the system $\Big{\{} \Omega^n , \mathcal{M}^n, r(\cdot), \gamma, T(\cdot) \Big{\}}$. \textcolor{Mycolor3}{ The goal of the centralized controller is to maximize the expected sum of discounted rewards (\ref{centralized problem - general problem}).} The expectation is computed over the joint PMF of the whole system trajectory $\bs(1), \bm(1),  ..., \bs(M), \bm(M)$ from time $t=1$ to $t=M$, where this joint probability mass function (PMF) is generated if agents follow policy $\pi(\cdot)$, eq. (\ref{centralized policy - equation}), for their action selections at all times and the initial state $\bs(1) \in \mathcal{S}$ is randomly selected by the initial distribution $\bs(1) \sim \alpha_\bs $. \textcolor{Mycolor3}{ For the sake of having a more compact notation to refer to the system trajectory, hereafter, we represent the realization of a system trajectory at time $t$ by $\rtr(t)$ which corresponds to the tuple $\langle \ro_1(t), ..., \ro_n(t), \rm_1(t), ..., \rm_n(t) \rangle$ and the realization of the whole system trajectory by $\{\rtr (t)\}_{t=1}^{t=M}$.} Accordingly, the problem boils down to a single agent problem which can be denoted by
{\small\textcolor{Mycolor3}{
\begin{align}\label{centralized problem - general problem}
&  \!\!\underset{\pi(\cdot)}{\text{max}}
& &  \!\!\!\mathbb{E}_{p_\pi\big(\{\rtr (t)\}_{t=1}^{t=M}\big)}                  \!\Big{\{}
\!\!\sum_{t=1}^M \gamma^{t-1} r\big(\bs(t),\bm(t)\big)
                \!\Big{\}}  
\end{align}}}
where the policy $\pi$ can be expressed as a CMF
{\small
\begin{align} \label{centralized policy - equation}
    \pi \Big( \rm(t) \Big{|} \rs(t)  \Big)  =  p \Big(\rm(t) \Big{|} \rs(t) \Big),
\end{align}}
 \hspace{-2.5mm}
 and $p_{\pi}\big(\rs(t+1)|\rs(t)\big)$ is the probability of transitioning from $\rs(t)$ to $\rs(t+1)$ when the joint action policy $\pi(\cdot)$ is executed by the central controller. Similarly, $p_\pi\big(\{\rtr (t)\}_{t=1}^{t=M}\big)$ is the joint PMF of $\rtr(1), \rtr(2),...,\rtr(M)$ when the joint action policy $\pi(\cdot)$ is followed by the central controller.
 
 On one hand, problem (\ref{centralized problem - general problem}) can be solved using single-agent Q-learning \cite{Suttonintroduction} and the solution $\pi^{*} (\cdot)$ obtained by Q-learning is guaranteed to be the optimal control policy, given some non-restricting conditions \cite{jaakkola1994convergence}. On the other hand, the use-cases of the centralized approach are limited to the applications in which there is a permanent communication link with an unlimited bit-budget between the agents and the controller. Whereas these conditions are not met in many remote applications, where there is no communication infrastructure to connect the agents to the central controller.
 
 Given sufficient training time, and channels with the sufficient rate of communication between the agents and the central controller, the centralized algorithm provides us with a performance upper bound in maximizing the objective function (\ref{centralized problem - general problem}). \textcolor{Mycolor3}{ Perfect communication between the central controller and distributed agents, however, may not exist due to the resource limitations of the telecommunication/communication network. Thus, the aim of this paper is to introduce decentralized approaches which are run over practical bit-budgeted communication channels, yet show comparable performance levels. In the distributed scenario, the agents do not communicate with a central controller, but the bit-budgeted communications are performed for inter-agent message exchange. The centralized problem can be presented by an MDP and be solved efficiently by a single agent reinforcement learning algorithm. As explained in the section \ref{subsect: contributions}, the decentralized problem is a more complicated/general form of Dec-POMDP, where we know that a Dec-POMDP is already much more complex than an MDP to solve \cite{oliehoek2007dec} - to see further insights about the significance and the applications of the decentralized problem see e.g., \cite{mostaani2021task}. }

\vspace{-3mm}
\subsection{Problem Statement}\label{problem statement subsection}
\vspace{-3mm}
Here we consider a scenario in which the same objective function explained in Eq. (\ref{centralized problem - general problem}) needs to be maximized by the multi-agent system in a decentralized fashion, Fig. \ref{decentralized problem - figure}. Namely, agents with partial observability can only select their own actions. To prevail over the limitations imposed by the local observability, agents are allowed to have direct (explicit) communications, and not indirect (implicit) communications \cite{larrousse2018coordination,heylighen2016stigmergy}. \textcolor{Mycolor3}{However, the communication is done through a bit-budgeted but reliable channel. The bit-budget of the channel is $R$-bits per time step. Equivalently, each agent $i$ at every time step $t$ produces and transmits a single digit communication message $\bc_i(t) \in \mathcal{C}$} such that
 {\small\begin{equation} \label{eq: bit-budget constraint}
 \textcolor{Mycolor3}{log_2|\mathcal{C}| \leq R},
 \end{equation}}
  i.e., the size of the code-books $\mathcal{C}$ for all agents is the same and is less than $2^R$. The communication message $\bc_i(t)$ produced by agent $i$ is broadcast and received every agent $j \in \mathcal{N}_{-i}$.
 It should be noted that the design of the channel coding is beyond the scope of this paper and the main focus is on the compression of agents' observations. In particular we consider \textcolor{Mycolor3}{$R$} to be time-invariant and to follow:
{\small\textcolor{Mycolor3}{
 \begin{equation}\label{observation structure, scenario 2}
 \text{}
      \textcolor{Mycolor3}{R} < \text{min} \,\, \big{\{}H\big(\bo_1(t)\big), ..., H\big(\bo_n(t)\big) \big{\}}.
\end{equation}}}
 \vspace{-0.0cm}
 \begin{figure}[t]
  \centering
      \includegraphics[width=0.400\textwidth]{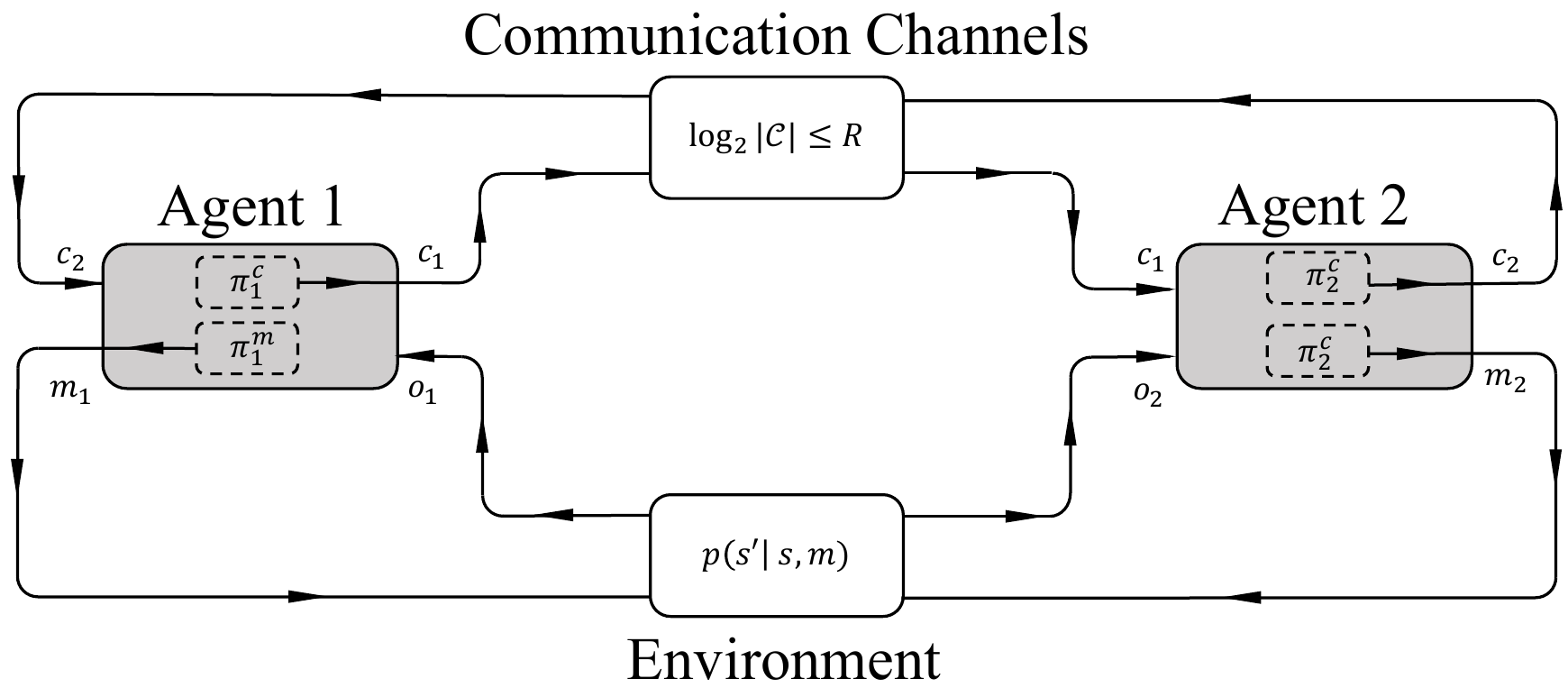}
      \vspace{-2mm}
  \caption{An illustration of the decentralized cooperative two-agent system with rate-limited inter-agent communications.}
  \label{decentralized problem - figure}
  \vspace{-4mm}
\end{figure}
 \textcolor{Mycolor3}{The above-mentioned information constraint which will be in place throughout this paper together with the observation structure assumed in eq. (\ref{eq: state vs observation}) are of the aspects that distinguish our work from many of the related works in the literature of multi-agent communications \cite{mostaani2019Learning,tung2021effective}.}
 Now let the function $\bg(t^{'})$ denote the system's \textit{return}:
 \vspace{-1mm}
{\small\begin{equation} \label{eq: cumulative rewards}
    \bg(t^{'})= {\sum}_{t=t^{'}}^{M}\gamma^{t-1} r\big(\bs(t),\bm(t)\big).
    \vspace{-1mm}
\end{equation}}
Note that $\bg(t^{'})$ is a random variable and a function of $t^{'}$ as well as the trajectory $\{\rtr (t)\}_{t=t^{'}}^{t=M}$. Due to the lack of space, here we drop a part of the arguments of this function. \textcolor{Mycolor3}{ In contrast to the centralized problem, the goal of the decentralized problem is to jointly design the communication/quantization as well as $\pi^c_i(\cdot)$ control policies $\pi^m_i(\cdot)$ for each agent $i \in \mathcal{N}$ to maximize the average return of the system.} \textcolor{Mycolor3}{ The control policy $\pi^m_i: \mathcal{M}\times \mathcal{C}^{n-1} \times  \Omega \rightarrow [0,1]$ of each agent $i$ is defined as CMF
\vspace{-1mm}
{\small\begin{align} \label{decentralized action policy}
    & \pi^m_i\Big(\rm_i(t) \Big{|} \ro_i(t),{\rc}_{-i}(t) \Big) = \notag \\
    & \mathrm{Pr} \Big(\bm_i(t)=\rm_i(t) \Big{|} \bo_i(t)=\ro_i(t),{\bc}_{-i}(t)={\rc}_{-i}(t) \Big),
\vspace{-2mm}
\end{align}}
in which, $\rc_{-i}(t) \in \mathcal{C}^{n-1}$ is a vector that includes all communication messages $\rc_j(t), \,\, \forall j \in \mathcal{N}_{-i}$. The communication policy $\pi^c_i :  \Omega \times \mathcal{C}^{n-1} \rightarrow \mathcal{C} $ of each agent $i$ is a deterministic data quantization (many to one) function:
\vspace{-1mm}
{\small\begin{align}\label{decentralized communication policy}
    & \rc_i(t) = \pi^c_i\Big( \ro_i(t),{\rc}_{-i}(t) \Big),
    \vspace{-1mm}
\end{align}}}
\textcolor{Mycolor3}{ which has a discrete domain $\Omega \times \mathcal{C}$, making the quantizer a discrete quantizer.} The joint control policy $\pi^m$ is a tuple made of $n$ elements with its $i$-th element being $\pi^m_i(\cdot)$. Similarly, The joint communication policy $\pi^c$ is another tuple with its $i$-th element being $\pi^c_i(\cdot)$.

According to the above definitions, the decentralized joint control and communication design problem is formalized as
{\small\begin{align}\label{decentralized problem - State Aggregation}
& \underset{\pi^m_i, \pi^c_i}{\text{max }} 
& & \mathbb{E}_{p_{\pi^m, \pi^c}
\big(\{\rtr (t)\}_{t=1}^{t=M}\big)}                  \Big{\{}
   \bg(1) 
\Big{\}},
  \; \; \; \; i \in \mathcal{N} \notag\\
& \text{s.t.} & &  
    \textcolor{Mycolor3}{log_2|\mathcal{C}| \leq R},
\end{align}}
\textcolor{Mycolor3}{  where the expectation is taken over ${p_{\pi^m, \pi^c}
\big(\{\rtr (t)\}_{t=1}^{t=M}\big)}$ which is the joint PMF of $\rtr(1), \rtr(2),\allowbreak ...,\rtr(M)$ when each agent $i \in \mathcal{N}$ follows the action policy $\pi^m_i(\cdot)$ and the communication policy $\pi^c_i(\cdot)$ and the initial state $\bs(1) \in \mathcal{S}$ is randomly selected by the initial distribution $\bs(1) \sim \alpha_\bs$.} \textcolor{Mycolor3}{ Given communication policy $\pi^c_i(\cdot), \,\, \forall i \in \mathcal{N}$, we now define the perception function $h_i(\cdot): \mathcal{S} \rightarrow \mathcal{C}^{n-1} \times \Omega$ of agent $i$ which is the lens through which agent $i$ perceives the state $\bs(t)$ of the environment.
\begin{align} \label{eq: perception}
 & h_i\big( \bs(t) \big) = \\ 
 & \langle \pi_1^c(\bo_1(t)), \pi_2^c(\bo_2(t)), ..., \bo_i(t), \pi_{i+1}^c(\bo_{i+1}(t)), ..., \pi_{n}^c(\bo_{n}(t))
 \rangle \notag
\end{align}
Agent $i$'s perception of the environment is characterized by the communication policy $\pi^c_j(\cdot)$ of each agent $j \in \mathcal{N}_{-i}$. Accordingly, agent $i$ uses its sensory signal $\bo_i(t)$ together with the received communication signals $\bc_{-i}(t)$ to acquire its perception of the environment. While the perception function defined here plays a role very similar to the observation function in Dec-POMDPs \cite{oliehoek2008optimal}, the main difference is that here we design communication policies such that they directly affect the perception of agents from the environment. In contrast, in the case of Dec-POMDPs, { the observation function is given}. Communication policies $\pi^c_j(\cdot), \forall j \in \mathcal{N}_{-i}$ partially define the perception function of agent $i$.}

To make the problem more concrete, further to (\ref{decentralized action policy}) and (\ref{decentralized communication policy}),  here we assume the presence of instantaneous and synchronous communications between agents, contrasting with the delayed \cite{mostaani2019Learning,oliehoek2016concise} and sequential communication models. \textcolor{Mycolor3}{ Fig. \ref{fig: Sequencial decision making} demonstrates this communication model during a single time-step.} As such, each agent $i$ at any time step $t$ prior to the selection of its action $\rm_i(t)$ receives \textcolor{Mycolor3}{ the communication vector ${\rc}_{-i}(t)$ that encodes the observations of each agent $j \in \mathcal{N}_{-i}$ at time $t$.}


\begin{figure}[htbp!]
  \centering 
      \includegraphics[width=0.480\textwidth]{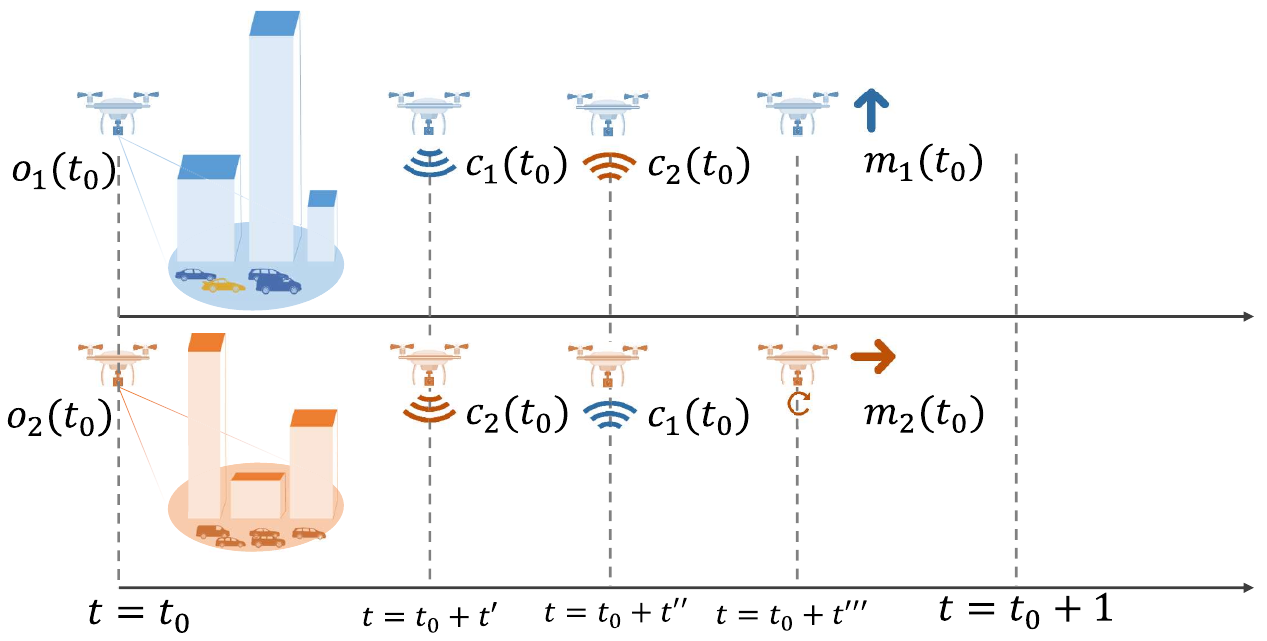} 
      \vspace{-2mm}
  \caption{Ordering of observation, communications and action selection for synchronous and instantaneous communication model in a multi-UAV object tracking example, with $0<t'<t''<t'''<1$. At time $t=t_0$ both agents (UAVs) make local observations on the environment. At time $t=t_0+t'$ both agents select a communication signal to be generated. At time $t=t_0+t''$ agents receive a communication signal from the other agent. At time $t=t_0+t'''$ agents select a domain level action, here it can be the movement of UAVs or rotation of their cameras etc.}
   \label{fig: Sequencial decision making}
   \vspace{-1mm}
\end{figure}

In a general approach, the selection of communication action $\rc_i(t)$ at agent $i$ could be conditioned on both $\ro_i(t)$ and ${\rc}_{-i}(t)$. \textcolor{Mycolor3}{ Since we assume instantaneous and synchronous inter-agent communications, here we are focused on communication policies of type $\pi^c_i\big(\ro_i(t) \big)$, where communication actions of each agent at each time are selected only based on its observation at that time. For clear reasons, it is not possible to adopt a synchronous and instantaneous inter-agent communication model and yet take the communication message $\rc_{-i}(t)$ into account when selecting the communication $\rc_i(t)$ at agent $i$.} Here we assume that the communication resources are split evenly amongst the agents, by considering the bit-budget of all communication channels to be equal to \textcolor{Mycolor3}{$R$}. As such, each agent $i \in \mathcal{N}$ encodes its observation $\ro_i(t)$ to $\rc_i(t)$ using a code-book $\mathcal{C}$ of the same length $|\mathcal{C}|$ - with the constraint (\ref{eq: bit-budget constraint}) in place.

\vspace{-3mm}
\section{State Aggregation for Information Compression (SAIC) in multi-agent Coordination Tasks}\label{sec: SAIC}
\vspace{-1mm}
The main result of this section - provided by Theorem \ref{The main theorem} - is to show that finding the quantization policy in the joint control and quantization problem (\ref{decentralized problem - State Aggregation}) can be approximated by a TODC problem. The goal of this problem is to quantize the observations of all agents according to how valuable these observations are within any specific task. The value of observations should be measured by the value function $V^*(\cdot)$ - eq. (\ref{eq: define value function}). Lemma \ref{lemma: quantization to k-median} approximates the TODC to a k-median clustering of the of observations according to their values, while lemma \ref{lemma: compute optimal value} computes the value function of each agent's observation. The concluding remarks of this section study the convergence and the optimality of the decentralized control policies.

 \begin{figure*}[htbp!] 
  \centering 
      \includegraphics[width=0.97\textwidth]{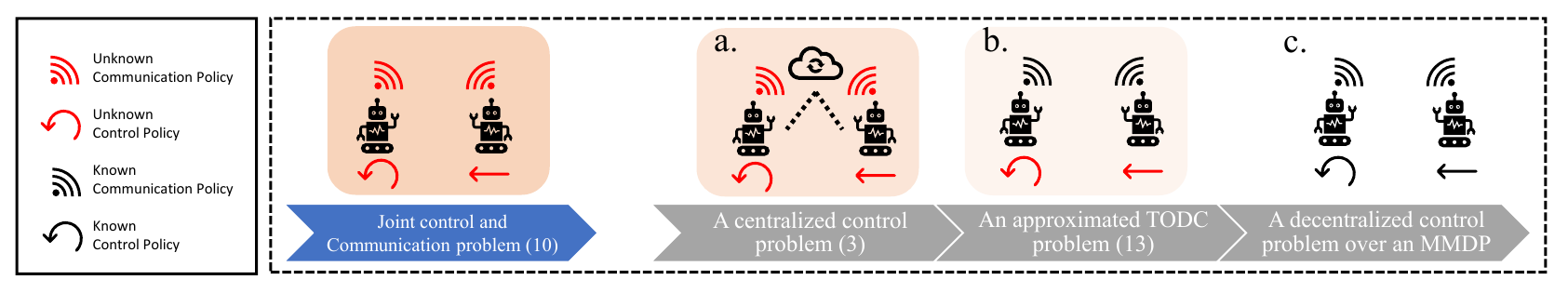} 
    \vspace{-0.2cm}
  \caption{ Here we show how we approached solving the joint control and communication problem for a distributed multi-agent system in a sequence of steps. According to the legend, one can understand that at the end of each step what are the known and unknown policies. a. This step solves the problem (\ref{centralized problem - general problem}) for a centralized multi-agent system where the
  objective is to design one centralized control strategy. b. This step solves the problem (\ref{k-median clustering problem - body}) for a distributed multi-agent system where the objective is to design the communication policies of all agents. c.  this step solves the problem for a distributed multi-agent system where the objective is to design the control policies of all agents.
 }
  \label{fig: joint problem break down}
\end{figure*}

 \textcolor{Mycolor3}{ Fig. \ref{fig: joint problem break down} is brought to demonstrate the chronological order according to which a joint communication and quantization is solved by SAIC. Our proposed scheme, SAIC, breaks down the joint communication and quantization problem to smaller problems that are feasible to solve.} In this section, however, the subsections are organized according to the logical order that these smaller problems are encountered: (A) in section \ref{subsect: TODQ problem} , we address the communication design of multi-agent communications by transforming the primary joint control and quantization problem (\ref{decentralized problem - State Aggregation}) to a novel problem (\ref{eq: generalized data quantization problem - body}) called TODC - step "b" of the Fig. \ref{fig: joint problem break down}. (B) Since solving the TODC problem relies on the knowledge of the value function $V^*(\cdot)$, it is necessary to obtain the value function $V^*(\cdot)$ prior to solving the TODC problem. In section \ref{subsec: centralized training}, the optimal value function $V^*(\cdot)$ is obtained via a centralized training phase - step "a" of the Fig. \ref{fig: joint problem break down}. Given the knowledge of the value function $V^*(\cdot)$, the TODC problem incorporates the features of the specific control task in the communication design problem. Accordingly, we can separately solve the communication problem with very little compromise on the optimality of the system's expected return. (C) As the final step, in section \ref{subsect: decentralized training phase}, decentralized training phase is carried out to distributively design the control policy of each agent given the communication/quantization policy obtained via solving the TODC problem. Decentralized training is shown in step "c" of the Fig. \ref{fig: joint problem break down}. Since we follow standard methods to carry out the centralized training - steps "a" of the  Fig. \ref{fig: joint problem break down} - we will be mainly focused on deriving and solving the TODC problem and providing guarantees on the performance of the MAS in the decentralized training phase - steps "b" and "c" of the Fig. \ref{fig: joint problem break down}  respectively. 
\textcolor{Mycolor3}{ Fig. \ref{fig: illustration of SAIC} illustrates how SAIC performs data compression while it maintains the performance of the multi-agent system in its task.}

 \begin{figure}[htbp!] 
  \centering 
      \includegraphics[width=0.450\textwidth]{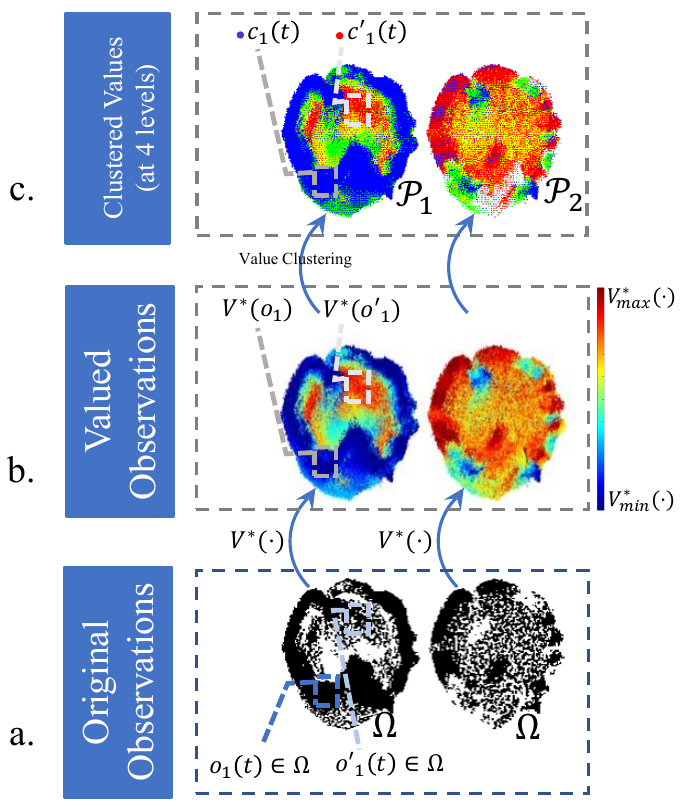} 
    \vspace{-0.2cm}
  \caption{\textcolor{Mycolor3}{ The subplots of this figure illustrate how in SAIC we transform a high-dimensional ($\sigma$-dimensions) and high-precision observation space into aggregated one-dimensional low-precision/digitized communication message space. This figure is plotted for a scenario where $R=2$ (bits per channel use) and thus, observation values are clustered at $2^R = 4$ different levels. a. A 2D demonstration of the original high-dimensional and high-precision observation space of agents is shown here in black and white. b. After carrying out the centralized training phase we will obtain the value function $V^*(\cdot)$ - which acts as indirect measure of the usefulness of observation data to be communicated. Now by applying the value function $V^*(\cdot)$ at every point of the original observation space we get valued observations - a one-dimensional high-precision space as the output space of the value function $V^*(\cdot)$. c. By clustering the observation points according to their corresponding values for each agent $i$ we would get a one-dimension and low-precision/digitized communication message space. The quantization illustrated in this diagram is using only 4 levels of quantization that are represented by 4 colours. All the points in the observation space of the agent $i$ which are represented by the same colour, in subplot c, will be represented by a unique communication message - i.e., the accuracy of the original data is reduced and hence requires fewer communication bits to be transmitted. Accordingly, agent 1, after observing $\ro_1(t)$ transmits the communication message $\rc_1(t)$ which is a compressed version of $\ro_1(t)$ while it maintains the performance of the multi-agent team in maximizing their expected return.}}
  \label{fig: illustration of SAIC}
\end{figure}

\vspace{-1mm}

\vspace{-0mm}
\subsection{Task-Oriented Data Compression Problem}\label{subsect: TODQ problem}
The main result of this section is provided by Theorem \ref{The main theorem}. \textcolor{Mycolor3}{ This theorem departs from the joint communication/quantization and control problem and arrives at the task-oriented data compression problem  (\ref{eq: generalized data quantization problem - body})}.
\begin{theorem}\label{The main theorem}
 The design of the communication policy in problem (\ref{decentralized problem - State Aggregation}) can be approximated as a generalized data quantization problem
 \textcolor{Mycolor3}{{\small\begin{align}\label{eq: generalized data quantization problem - body}
&  \underset{\pi^c_i(\cdot)}{\text{min}}
& & \!\!\mathbb{E}_{p_{\pi^m,\pi^c}
\big( h_i\big( \rs(1) \big) \big)}
\Big{\{}
 \big{|} V^{*}\big( \bs(1) \big) - V^{*}\big( h_i\big( \bs(1) \big) \big)  \big{|}          \Big{\}} \notag \\
& \text{s.t.} 
& &  \!\! \textcolor{Mycolor3}{log_2|\mathcal{C}| \leq R},
\end{align}}}
in which the measure of distortion is the absolute difference of the value functions $V^{*}\big(\bs(t)\big)$ and $V^{*}\big( h_i \big( \bs(1) \big) \big)$ with the source of information $\bs(t) \in \Omega^n$ being a Markovian stochastic process. \textcolor{Mycolor3}{ The function $V^{*}\big( h_i\big( \bs(1) \big) \big) $ measures the optimal value of the perceived state $h_i\big( \bs(1) \big)$ from agent $i$'s perspective.}
\end{theorem}

\vspace{-2mm}
\begin{proof}
    Appendix \ref{appendix: proof of the main theorem}.
\end{proof}

\textcolor{Mycolor3}{ In Appendix \ref{subsect: proof: lem: optimal value of the aggregated state}, we provide more details on how to obtain the value  $V^{*}\big( h_i\big( \bs(1) \big) \big) $ of the perceived state from the agent $i$'s point of via Lemma \ref{lem: optimal value of the aggregated state}. This value function allows us to indirectly quantify the usefulness of agent $i$'s observation. With this interpretation in mind, in the TODC problem (\ref{eq: generalized data quantization problem - body}), unlike conventional quantization problems, we are not minimizing the absolute difference between the original signal $\bs(1)$ and its quantized version $h_i\big( \bs(1) \big)$. Instead, we are minimizing the distance between how useful/valuable the original signal $\bs(1)$ is and how useful the quantized version of the signal $h_i\big( \bs(1) \big)$ are for the task at hand. This is in-line with what many believe as the mission of the goal-oriented/task-oriented communications. Let us recall that the value function here is an \emph{indirect} measure of usefulness, as it can be obtained for any task that can be expressed via Markov Decision Processes - making it a measure of usefulness that is applicable to a plethora of scenarios \cite{mostaani2021task,azari2022evolution}.}

The significance of the result obtained by Theorem \ref{The main theorem} is multi-fold: (i) Multi-dimensional observations will be transformed to one-dimensional output space of the value functions, reducing the complexity of the clustering algorithm, (ii) It can be shown that the observation points will be linearly separable when being clustered according to the problem (\ref{eq: generalized data quantization problem - body}), (iii) It is widely accepted that the mission of goal oriented communications is to incorporate the usefulness/value of the data for the task when designing the task-effective communications. The result of Theorem \ref{The main theorem}, in which the design of the quantizer relies on the value/usefulness of observations resonates well with this purpose of goal-oriented communications. (iv) It is known that the value of observations starts to grow as we get closer to the ultimate target of the task in hand. With this interpretation of "target" in mind, the finding of Theorem \ref{The main theorem} is in line with the adaptive quantization schemes, which stretch the quantization intervals when the observations are far from the target and sharpen the quantization when the observations are closer to the target \cite{nair2004stabilizability, yuksel2013jointly}. This interpretation is also confirmed by our numerical results in section \ref{Numerical results - Section}, Fig. \ref{fig: state-aggregation - centralized learning}.

To solve a quantization problem as (\ref{eq: generalized data quantization problem - body}) using non-variational techniques, it is customary to approximate/convert a quantization problem by/to a clustering problem \cite{lloyd1982least,linde1980algorithm}. Lemma \ref{lemma: quantization to k-median} approximates the quantization problem (\ref{eq: generalized data quantization problem - body}) by a clustering problem.

\begin{lemma} \label{lemma: quantization to k-median}
 The quantization problem (\ref{eq: generalized data quantization problem - body}) can be \textcolor{Mycolor3}{approximated by a clustering problem}
     \begin{equation}\label{k-median clustering problem - body}
        \begin{aligned}
        &  \underset{\mathcal{P}_i}{\text{min}}
        & & {\sum}_{k=1}^{{|\mathcal{C}|}} {\sum}_{\ro_i(t) \in \mathcal{P}_{i,k}} \Big{|}                 V^{*}\big(\ro_i(t)\big) - \mu^{'}_k \Big{|}, 
        \end{aligned}
    \end{equation}
    where $\mu'_k$ is the centroid of the $k$-th cluster $\mathcal{P}_{i,k}$ and $\mathcal{P}_i = \{ \mathcal{P}_{i,1}, \dots, \mathcal{P}_{i,{|\mathcal{C}|}} \}$ is a partition of the observation space $\Omega$. Similar to any other quantization function, the quantizer $\pi^c_i(\cdot)$, can be uniquely described by the partition $\mathcal{P}_i$ together with $\mathcal{C}$.
\end{lemma}
\begin{proof}
Appendix \ref{appendix: proof: lemma: quantization to k-median} provides proof and discussions.
\end{proof}

\textcolor{Mycolor3}{ The problem (\ref{k-median clustering problem - body}), can be solved via k-median clustering. In order to that, one can first perform the k-median clustering on the observation values by solving
     \begin{equation}
        \begin{aligned}
        &  \underset{\mathcal{V}_i}{\text{min}}
        & & {\sum}_{k=1}^{2^B} {\sum}_{V^*(\ro_i(t)) \in \mathcal{V}_{i,k}} \Big{|}                 V^{*}\big(\ro_i(t)\big) - \mu^{''}_k \Big{|}, \notag
        \end{aligned}
    \end{equation}
where $\mathcal{V}_i$ is the set of all observation values of agent $i$ and $\{\mathcal{V}_{i,1}, ..., \mathcal{V}_{i,|\mathcal{C}|}\}$ is its partition. Afterwards, as shown in Figure \ref{fig: illustration of SAIC}, the observation points should be clustered according to the clustering of their corresponding values. That is, any two distinct observation points $\ro'_i, \ro''_i \in \Omega$ are clustered together in $\mathcal{P}_{i,j}$ if and only if their values $V^*(\ro'_i), V^*(\ro''_i) \in \mathcal{V}_{i,j}$ are in the same cluster $\mathcal{V}_{i,j}$.} 

Theorem \ref{The main theorem} together with lemma \ref{lemma: quantization to k-median} allows us to find a communication/quantization policy $\pi^{c^{}}_i(\cdot)$ by clustering the input space $\Omega$ of the communication policy according to the values $V^{*}\big( \ro_i(t)\big)$ of the input points. The performance guarantees for the obtained communication/quantization policy will be shown in section \ref{sec: error bound}. One can obtain $V^{*}\big( \ro_i(t)\big)$ via solving the centralized problem (\ref{centralized problem - general problem}) by Q-learning. The subsection \ref{subsec: centralized training}, details a centralized training approach for obtaining the value observations $V^{*}\big( \ro_i(t)\big)$. 

\subsection{Centralized Training Phase} \label{subsec: centralized training}

While solving the TODC problem can provide us with a task-effective design of quantization policy, to solve (\ref{eq: generalized data quantization problem - body}) we need to know the value of observations according to the optimal centralized control policy. By solving the centralized problem (\ref{centralized problem - general problem}), the value of joint observations and actions $Q^{*}\big(\rs(t),\rm(t)\big)$ can be obtained. \textcolor{Mycolor3}{ Let us recall that the centralized training phase will only yield an optimal policy if the environment is jointly observable - as described by condition \ref{condition: joint observability}.}
\begin{condition} \label{condition: joint observability}
    \begin{align} \label{joint observability}
     \rs(t) = \langle \ro_1(t),..., \ro_n(t) \rangle.
  \end{align}
\end{condition}
Accordingly, following the lemma \ref{lemma: compute optimal value} we can compute the value of each agent's observations $V^{*}\big( \ro_i(1)\big)$. But before lemma \ref{lemma: compute optimal value},  let us first give an intuitive/philosophical meaning of the centralized training and distributed execution. We know that in task-oriented communication design, our goal is to take into account the usefulness/value of the data for the task in hand. Thus we need to first be able to measure the usefulness/value of the data to be transmitted. The centralized training phase is needed to come up with a precise measure of usefulness for the specific task in hand. We have already shown in Theorem \ref{The main theorem}, that this measure of usefulness is nothing but the value observations $V^{*}\big( \ro_i(1)\big)$ - yet the exact function values can be known only after the centralized training phase. During the centralized training phase, we assume perfect communication between all agents and a central controller - this is a common practice in the literature of multi-agent communications and coordination \cite{FoersterLearning, FoersterCounter}. Whereas, in the decentralized training - step "c" of the Fig. \ref{fig: joint problem break down} - as well as in the execution phase, we assume bit-budgeted communications. That is, all the results reported for SAIC in section \ref{Numerical results - Section} are obtained via bit-budgeted communications.

\textcolor{Mycolor3}{ \begin{lemma} \label{lemma: compute optimal value}
 One can compute the $V^{*}\big( \ro_i(1)\big)$ following
 {\small\begin{align} \label{value function iterated expectation simplified - State aggregation - body}
   V^{*}\big(\ro_i(t)\big) = 
   \sum_{\ro_{-i}(t) \in \Omega^{n-1}}    \underset{\rm}{\text{max }} \; Q^{*}\big(\rs(t),\rm(t)\big)   p\big(\bo_{-i}(t) = \ro_{-i}(t)\big).       \notag
\end{align}}
\end{lemma}}
\begin{proof}
    Appendix \ref{appendix: proof: lemma: compute optimal value}.
\end{proof}
Based on (\ref{value function iterated expectation simplified - State aggregation - body}), $V^{*}\big(\ro_i(1)\big)$ can be computed both analytically (if transition probabilities of environment are available) and numerically. \textcolor{Mycolor3}{ As detailed in Algorithm 1, SAIC first solves a centralized control problem to compute the value $V^{*}(\ro)$ for all $\ro \in \Omega$ - this is equivalent to the step "a" of the Fig. \ref{fig: joint problem break down} and subplot (b) of the Fig. \ref{fig: illustration of SAIC}. Afterwards, SAIC solves the approximated TODC problem (\ref{eq: generalized data quantization problem - body}) by converting it to a k-median clustering (\ref{k-median clustering problem - body}), leading to an observation aggregation/quantization function for each agent $i$ determined by  $\pi^{c}_i(\cdot)$ - this is equivalent to the step "b" of the Fig. \ref{fig: joint problem break down} and the subplot (c) of the Fig. \ref{fig: illustration of SAIC}.} By following this aggregation function, the observations $\ro_i(t) \in \Omega$ will be aggregated/quantized such that the performance of the multi-agent system in terms of the objective function it attains is optimized. As SAIC uses a deterministic mapping of observation $\ro_i$ to produce the communication message $\rc_i$, SAIC is guaranteed to have positive signalling \cite{lowe2019pitfalls}.

\subsection{Obtaining Decentralized Control Policies via a Decentralized Training Phase} \label{subsect: decentralized training phase}

\textcolor{Mycolor3}{ Upon the availability of the $\pi^{c}_i(\cdot), \,\, \forall i \in \mathcal{N}$, which was obtained by solving problem (\ref{k-median clustering problem - body}), we need to find control policies for all agents corresponding to the communication policies $\pi^{c}_i(\cdot), \,\, i \in \mathcal{N}$. That is, we now solve the problem (\ref{decentralized problem - State Aggregation}) by plugging the exact communication policy $\pi^{c}_i(\cdot) \,\, \forall i \in \mathcal{N}$ into it. Within this training phase - referred to as the decentralized training phase - control $Q$-tables $Q^m_i(\cdot) \,\, \forall i \in \mathcal{N}$ are obtained - step "c" of the Fig. \ref{fig: joint problem break down}. This training phase, as well as the execution phase of the algorithm, can both be carried out distributively, while agents communicate over bit-budgeted channels using the communication policies obtained before in section \ref{subsect: TODQ problem}. The following remarks are brought to characterize the performance of SAIC, in the decentralized training phase.}

\textcolor{Mycolor3}{ We now first define the concept of lumpablity, according to which we will then set a condition - Condition \ref{condition: lumpability} - for the correctness of remarks 3 and 4.
}

\textcolor{Mycolor3}{
\begin{definition}\textbf{Lumpability of an MDP:} \label{lumpablity}
Let $\alpha_{\bs}$ be the probability distribution of the { initial state} of an MDP at the initial step. The MDP is called (strongly) lumpable with respect to the perception function $h_i(\cdot)$ if the the transitions between all the perceived states $h_i(\bs(t))$ - which are perceived through the lens of $h_i(\cdot)$ - follow Markov rule for every probability distribution $\alpha_{\bs}$ of the initial state of the original MDP \cite{rubino1989weak}. 
\end{definition}}

 \begin{condition} \label{condition: lumpability}
      Let the environment as perceived from the perspective of agent $i$ within the decentralized training phase be called an aggregated MDP denoted by $\Big{\{} \Omega \times \mathcal{C}^{n-1} , \mathcal{M}, r(\cdot), \gamma, T'(\cdot) \Big{\}}$, whereas the state space of the aggregated MDP $\Omega \times \mathcal{C}^{n-1}$ is an image of $\Omega^n$ under the perception function $h_i(\cdot)$. Now given the definition \ref{lumpablity}, assuming the lumpability of the underlying MDP $\Big{\{} \Omega^n , \mathcal{M}^n, r(\cdot), \gamma, T(\cdot) \Big{\}}$ with respect to $h_i(\cdot)$ is equivalent to the assumption that the aggregated $\Big{\{} \Omega \times \mathcal{C}^{n-1} , \mathcal{M}, r(\cdot), \gamma, T(\cdot) \Big{\}}$ is an MDP under every possible $\alpha_\bs$. This assumption is in place for the correctness of remarks 3 and 4.
 \end{condition}


\begin{algorithm}\label{SAIC - Algorithm}
\caption{ State Aggregation for Information Compression (SAIC)}
\begin{algorithmic}[1]

\State \small \textbf{Input:} $\gamma$, $\alpha$, $c$
 \State \textbf{Initialize} all-zero table $N^m_{i}\big(\ro_i(t), {\rc}_{-i}(t),\rm_i(t)\big)$, for $i \in \mathcal{N}$ 
 \State $\;\;\;\;\;\;\;\;\;\;\;\;\;\;\!$ and Q-table $Q^{m}_{i}(\cdot) \leftarrow Q^{m,(k-1)}_{i}(\cdot)$, for $i \in \mathcal{N}$
 \State $\;\;\;\;\;\;\;\;\;\;\;\;\;\;\!$ and all-zero Q-table $Q\big(\ro_i(t),\ro_j(t),\rm_i(t),\rm_j(t)\big)$.
 \State Obtain $\pi^{*}(\cdot) \text{ and } Q^{*}(\cdot)$ by solving (\ref{centralized problem - general problem}) using Q-learning \cite{Suttonintroduction}. 
 \State Compute $V^{*}\big( \ro_i(t) \big)$ following eq. (\ref{value function iterated expectation simplified - State aggregation - body}), for $\forall \ro_i(t) \in \Omega$.
 \State Solve problem (\ref{k-median clustering problem - body}) by applying k-median clustering to obtain $\pi^c_i(\cdot)$, for $i \in \mathcal{N}$.
  \For{each episode $k=1:K$} 
 \State Randomly initialize local observation $\ro_i(t=1)$, for $i \in \mathcal{N}$
\For{$t_k = 1:M$}

\vspace{1mm}          
            \State Select $\rc_i(t)$ following $\pi^{c}_i(\cdot)$, for $i \in \mathcal{N}$

            \State Obtain message $ {\rc}_{-i}(t)$, for $i \in \mathcal{N}$
            
            \State Update $Q^{m}_{i}\big(\ro_i(t-1), {\rc}_{-i}(t-1),\rm_i(t-1)\big)$
            , for $i \in \mathcal{N}$
            
            \State Select $\rm_i(t) \in \mathcal{M}$ following UCB, for $i \in \mathcal{N}$
            \State Increment $N^m_{i}\big(\ro_i(t), {\rc}_{-i}(t),\rm_i(t)\big) $, {for} $i \in \mathcal{N}$
            \State Obtain reward $r\big( \rs(t),\rm(t) \big)$, {for} $i \in \mathcal{N}$
            
            \State Make a local observation $\ro_i(t)$, for $i \in \mathcal{N}$
    \vspace{1mm}

\State $t_k=t_k+1$

\EndFor\label{euclidendwhile-2}
\State \textbf{end}
\State Compute $\sum^{M}_{t=1} \gamma^t r_t$ for the $l$th episode
\EndFor
\State \textbf{end}
\vspace{1mm}

\State \textbf{Output:} $Q^{m}_{i}(\cdot)$, 
\State $\;\;\;\;\;\;\;\;\;\;\;\;\;\;\;\!$and $\pi_{i}^{m}\big(\rm_i(t)|\ro_i(t), {\rc}_{-i}(t)\big)$ by following greedy policy
{for} $i \in \mathcal{N}$

%
\vspace{-1mm}
\end{algorithmic}
\end{algorithm}


\textcolor{Mycolor3}{\emph{Remark 1:} The optimal policy $\pi^*(\cdot)$ is achievable by the centralized training phase. Assuming Condition \ref{condition: joint observability} to hold, the environment is fully observable for the central controller while the central controller posses the ability to jointly select the actions for all agents. The problem will thus reduce to a single agent Q-learning applied on an MDP with asymptotic convergence to the optimal policy $\pi^*(\cdot)$.}

\textcolor{Mycolor3}{
\emph{Remark 2:} During the decentralized training phase, each agent, instead of viewing the environment as the original underlying MDP denoted by $\Big{\{} \Omega^n , \mathcal{M}^n, r(\cdot), \gamma, T'(\cdot) \Big{\}}$, views an aggregated form of the original MDP denoted by $\Big{\{} \Omega \times \mathcal{C}^{n-1} , \mathcal{M}, r(\cdot), \gamma, T'(\cdot) \Big{\}}$. The aggregated MDP experienced by agent $i$ will be an MDP itself, if the conditions \ref{condition: joint observability} and \ref{condition: lumpability} hold.}

\textcolor{Mycolor3}{
\emph{Remark 3:} The MAS, during the decentralized training phase, will be composed of $n$ different MDPs with identical state space $\Omega \times \mathcal{C}^{n-1}$, action space $\mathcal{M}$ and reward signal. The resulting multi-agent environment will be, according to the definition, a multi-agent MDP (MMDP)  \cite{boutilier1999multiagent}. }

\emph{Remark 4:} Within the distributed training phase, distributed Q-learning is applied to a deterministic MMDP \footnote{The definition of MMDP in \cite{boutilier1999multiagent} is identical to the definition of cooperative MAMDP used in \cite{lauer2000distributedQ}.}, which leads to an asymptotically optimal control policy \cite{lauer2000distributedQ} \footnote{This training phase can result in an asymptotically optimal control policy of all agents for non-deterministic MMDPs. This, however, will require $n$ additional centralized training phases prior to the decentralized training phase, where $n$ is the number of agents.}. For this remark to be true conditions \ref{condition: joint observability} and \ref{condition: lumpability} must hold.

\textcolor{Mycolor3}{
Note that the control policy $\pi^{m,SAIC}_i(\cdot)$ that is obtained within the distributed training phase of SAIC is optimal for the given communication policy $\pi^{c, SAIC}(\cdot)$, that was obtained within the centralized training phase. Therefore, $\pi^{m,SAIC}_i(\cdot)$ is not necessarily an optimal solution to the problem (\ref{decentralized problem - State Aggregation}). In Theorem section \ref{sec: error bound}, however, we set an upper-bound on the possible loss on the expected return of the system due to the joint selection of $\pi^{m,SAIC}_i(\cdot)$ and $\pi^{c, SAIC}(\cdot)$.}

\section{Characterizing the error bound of SAIC} \label{sec: error bound}
\textcolor{Mycolor3}{
 As discussed in section \ref{sec: SAIC}, SAIC uses two approximations to solve the original joint quantization and control problem. It was not, however, explained that how these approximation would impact the performance of SAIC in terms of the system's average return. By extending the results of \cite{abel2016near} to a multi-agent scenario, we characterize the performance gap of SAIC proposed in section \ref{sec: SAIC}. Instead of measuring the difference between the average return obtained by SAIC with that of the jointly optimal policies for the problem (\ref{decentralized problem - State Aggregation}), in Theorem \ref{theorem: error bound}, we measure the performance gap between the average return attained by SAIC with that of the centralized controller - whereas the latter has had access to perfect communications and as well as full observability of the environment. The measured gap is, indeed, larger than the performance gap between SAIC and a hypothetical jointly optimal solution to (\ref{decentralized problem - State Aggregation}), as in the case of the central controller there is no communication/observation limitation in place. The performance gap between SAIC and the centralized solution provided by Theorem \ref{theorem: error bound} is proposed in terms of the discount factor $\lambda$ of the task and a positive scalar $\epsilon$. Definition \ref{def: cost-uniform} details the notion of $\epsilon$-cost uniform. Lemma \ref{lem: compute epsilon} is proposed to compute the value of $\epsilon$ for SAIC. } 

\textcolor{Mycolor3}{
\begin{definition}\label{def: cost-uniform}
Given a positive number $\epsilon$ a subset $\mathcal{P}_{i,k} \subset \Omega$ is said to
be $\epsilon$-cost-uniform with respect to the policy $\pi(\cdot)$ if the following conditions hold for
two arbitrary observations $\ro' , \ro'' \in \mathcal{P}_{i,k}$:
{\small\begin{align}
        & c_1: & &  \mathcal{M}_\pi(\ro' ) = \mathcal{M}_\pi(\ro'') \\
        & c_2: & &  \text{For any } \rm \in \mathcal{M}_\pi(\ro' ) : | Q^{\pi}(\ro', \rm) - Q^{\pi}(\ro'', \rm)| < \epsilon,
\end{align}}
where $\mathcal{M}_\pi(\ro' ) = \big{\{} \rm \in \mathcal{M} :   \pi(\rm | \ro') > 0 \big{\}}$.
\end{definition} }
\vspace{-2mm}

\textcolor{Mycolor3}{
\begin{theorem}\label{theorem: error bound}
 Consider a multi-agent system in which agents are subject to local observability and local action selection. If agents are allowed to communicate through communication channels with a bit-budget $R$-bits at each time step, the maximum achievable expected return of the multi-agent system following SAIC algorithm will be in a small neighbourhood of the same MAS if it was controlled with a centralized unit under perfect communications:
    {\small \begin{align}
        & \mathbb{E}_{p_{\pi^*}(\{\rtr (t)\}_{t=t_0}^{t=M})}\big{\{}  \bg(t_0)  \big{\}} - \mathbb{E}_{p_{\pi^m_i,\pi^c_i}(\{\rtr (t)\}_{t=t_0}^{t=M})}\big{\{}  \bg(t_0)  \big{\}} < \notag \\
        & \frac{2 \,\epsilon }{(1-\gamma)^2},
    \end{align}}
    where $\gamma$ is the discount factor and $\epsilon$ should be computed according to lemma \ref{lem: compute epsilon}, conditioned on the lumpability of the original MDP - Condition \ref{condition: lumpability}.
\end{theorem}
}
\begin{proof}
    Appendix \ref{appendix: proof: theorem: error bound}.
\end{proof}

In Theorem \ref{theorem: error bound}, we will show that the error gap between
\vspace{-1mm}
\textcolor{Mycolor3}{
\begin{lemma} \label{lem: compute epsilon}
 Given the partition $\mathcal{P}_i = \{\mathcal{P}_{i,1}, ..., \mathcal{P}_{i,2^R}\} $ that is obtained by solving eq. (\ref{k-median main - State Aggregation}) during the centralized training phase, all subsets $\mathcal{P}_{i,k} $ for $k \in \{1,2,..., 2^R\}$ are  $\epsilon$-cost-uniform with respect to the optimal joint policy $\pi^*(\cdot)$ where $\epsilon$ can be obtained by the following
    {\small\begin{equation}
        \epsilon/2 = \underset{k,\ro_i}{\text{max}} \,\,  \Big{|} V^{*}\big(\ro_i(t)\big) - \mu^{'}_k \Big{|}.
    \end{equation}}
\end{lemma}
\begin{proof}
 Following definition \ref{def: cost-uniform} and eq. (\ref{k-median clustering problem - body}) the proof is straightforward.
\end{proof}}

\vspace{-5mm}
\section{Performance Evaluation} \label{Numerical results - Section}
\vspace{-2mm}
\textcolor{Mycolor3}{In this section, we evaluate our proposed schemes via numerical results for a particular geometric consensus problem with finite observability called the rendezvous problem. Geometric consensus problems arise in numerous emerging applications such as UAV/vehicle platooning - making them a meaningful application area for the framework proposed by this paper \cite{barel2017come}. The numerical results achieved by SAIC will prove the suitability of the proposed framework as a potential enabling technology for vehicle/UAV platooning under limited communications.}

\textcolor{Mycolor3}{The rendezvous problem, which is a sub-category of the geometric consensus,} has been previously investigated in the literature \cite{zilber2001communication,amato2009incremental}, whereas in our case the inter-agent communication channel is set to have a limited bit-budget. The rendezvous problem is of particular interest to us, also because it allows us to consider a cooperative MAS comprising of multiple agents that are required to communicate for their coordination task. In particular, as detailed in subsection \ref{rendezvous problem - subsection}, if the communication between agents is not efficient, at any time step $t$ each agent $i$ will only have access to its local observation $\ro_i(t)$, which is its own location in the case of rendezvous problem. This mere information is insufficient for an agent to attain the larger reward $C_2$, but is sufficient to attain the smaller reward $C_1$. Accordingly, compared with cases in which no communication between agents is present, in the set up of the rendezvous problem, efficient communication policies can increase the attained objective function of the MAS up to six-folds, as will be seen in Fig. 4. The system operates in discrete time, with agents taking actions and communicating in each time step $t=1,2,...$ . We consider a variety of grid worlds with different size values $N$ and different locations for the goal-point $\omega^T$. We compare the proposed SAIC and LBIC with (i) the centralized Q-learning scheme and (ii) the Conventional Information Compression (CIC) scheme which is explained in subsection \ref{conventional communication - subsection}. \textcolor{Mycolor3}{ Changing the reward function can also build new scenarios. For example, a reward function that encourages the agents to come together as close as possible but not collide with each other can emulate a vehicle platooning scenario. While useful, it is outside the scope of our work to investigate the response of the multi-agent system to different rewarding schemes.}  \textcolor{Mycolor3}{Note that, according to Theorem \ref{The main theorem}, regardless of the definition of the reward function, the geometric consensus problem (or in general the joint quantization and control problem) can be solved by SAIC if the necessary Conditions \ref{condition: joint observability} and \ref{condition: lumpability} are met, and centralized training phase is feasible. As the number of agents $n$ increases, the Q-learning for the centralized training phase becomes increasingly demanding in terms of computational complexity; this is where SAIC's bottleneck lies.}

 \subsection{Rendezvous Problem} \label{rendezvous problem - subsection}
  \begin{figure}[thb!]
  \centering
      \includegraphics[width=0.47\textwidth]{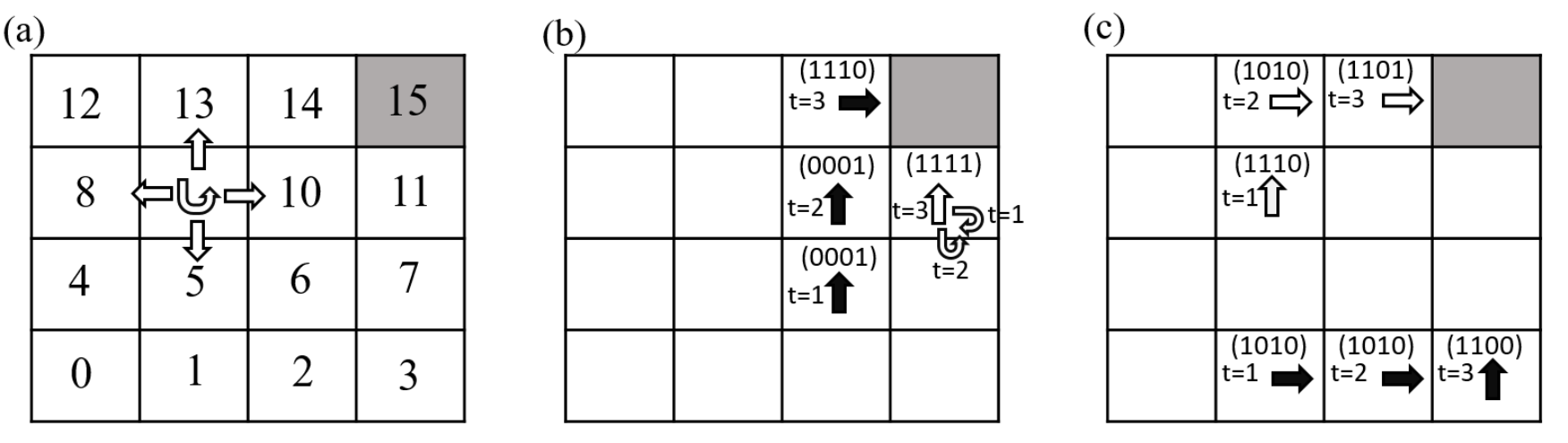}
  \caption{The rendezvous problem when $n=2$, $N=4$ and $\omega^T=15$: (a) illustration of the observation space, $ \Omega$, i.e., the location on the grid, and the environment action space $\mathcal{M}$, denoted by arrows, and of the goal state $\omega^T$, marked with gray background; (b) demonstration of a sampled episode, where arrows show the environment actions taken by the agents (empty arrows: actions of agent 1, solid arrows: actions of agent 2) and the $B=4$ bits represent the message sent by each agent. A larger reward $C_2> C_1$ is given to both agents when they enter the goal point at the same time, as in the example; (c) in contrast, $C_1$ is the reward accrued by agents when only one agent enters the goal position \cite{mostaani2019Learning}.}
  \label{fig: rendezvous problem}
\end{figure}

 As illustrated in Fig. \ref{fig: rendezvous problem}, in a rendezvous problem, multiple agents operate on an $N \times N$ grid world and aim at arriving at the same time at the goal point on the grid. Each agent $i \in \mathcal{N}$ at any time step $t$ can only observe its own location $\ro_i(t) \in \Omega$ on the grid, where the observation space is $\Omega = \{0,1,...,n^2-1\}$. Each episode terminates as soon as an agent or more visit the goal point which is denoted as $\omega^T \in \Omega$. That is, at any time step $t$ that the observation of each agent $i \in \mathcal{N}$ is a member of $\Omega^T$, the episode will be terminated - so the time horizon $M$ is non-deterministic. The subset $\mathcal{S}^T \subset \mathcal{S}$ also defines all state realizations where one or more agents are in the goal location i.e.,
 
 $\mathcal{S}^T = \{ \langle \ro_1(t), ..., \ro_n(t) \rangle \in \mathcal{S} \, | \,  \exists i \in \mathcal{N} : \ro_i(t) \in \omega^T \}$.
 
 We also define the subset $\mathcal{S}^T_{n'} \subset \mathcal{S}^T$ that includes all the terminal states where only $n'$ number of agents have arrived at the goal location i.e.,
 
  $\mathcal{S}^T_{n'} = \{ \langle \ro_1(t), ..., \ro_n(t) \rangle \in \mathcal{S} \, | \,  \forall i \in \mathcal{N}' : \ro_i(t) \in \omega^T \}$,
  
  where $\mathcal{N}' \subseteq \mathcal{N}$ is a subset of all agents with size $| \mathcal{N}' | = n'$. Following the same definition for $\mathcal{S}^T_{n'}$, the subset $\mathcal{S}^T_n$ is equivalent to the set of all terminal states where all agents are at the goal location. At time $t=1$, the initial position of all agents is randomly and uniformly selected amongst the non-goal states, i.e., for each agent $i \in \mathcal{N}$ the initial position of the agent is $\ro_i(1) \in \Omega - \{ \omega^T \}$.

At any time step $t=1,2,...$ each agent $i$ observes its position, or environment state, and acquires information about the position of the other agents by receiving a communication message vector $ {\rc}_{-i}(t)$ sent by the other agents $j \in \mathcal{N}_{-i}$ at the time step $t$.  Based on this information, agent $i$ selects its environment action $\rm_i(t)$ from the set $\mathcal{M} = \{\text{Right},\text{Left},\text{Up},\text{Down},\text{Stop}\}$, where an action $\rm_i(t) \in \mathcal{M}$ represent the horizontal/vertical move of agent $i$ on the grid at time step $t$. For instance, if an agent $i$ is on a grid-world as depicted on Fig. \ref{fig: rendezvous problem} (a), and observes $\ro_i(t)=4$ and selects "Up" as its action, the agent's observation at the next time step will be $\ro_i(t+1)=8$. If the position to which the agent should be moved is outside the grid, the environment is assumed to keep the agent in its current position. We assume that all these deterministic state transitions are captured by $T\big(\ro_1(t), ..., \ro_n(t),\rm_1(t), ..., \rm_n(t)\big)$, which can determine the observations of agents in the next time step $t+1$ following
{\small\[
\langle \ro_1(t+1), ..., \ro_n(t+1) \rangle
= T\big(\ro_1(t), ..., \ro_n(t),\rm_1(t), ..., \rm_n(t)\big).
\]}
Accordingly, given observations $\langle \ro_i(t+1), ..., \ro_n(t+1) \rangle$ and actions $ \langle \rm_1(t+1), ..., \rm_n(t+1) \rangle $, all agents receive a single team reward
{\small\begin{equation}\label{example environment noise distribution}
   r\big( \ro_1(t), ..., \ro_n(t),\rm_1(t), ..., \rm_n(t) \big)=
   \begin{cases}
    C_1, & \text{if  $P_1$}\\
    C_2, & \text{if  $P_2$},\\
    0, & \text{otherwise},\\
   \end{cases}
\end{equation}}
where $C_1 < C_2$ and the propositions $P_1$ and $P_2$ are defined as $P_1: T\big(\ro_1(t), ..., \ro_n(t),\rm_1(t), ..., \rm_n(t)\big) \in \mathcal{S}^T - \mathcal{S}^T_n$ and $P_2: T\big(\ro_1(t), ..., \ro_n(t),\rm_1(t), ..., \rm_n(t)\big) \in \mathcal{S}^T_n$. When only a subset $\mathcal{N}',\,\, |\mathcal{N}| = n' < n$ of agent arrives at the target point $\omega^T$, the episode will be terminated with the smaller reward $C_1$ being obtained, while the larger reward $C_2$ is attained only when all agents visit the goal point at the same time. Note that this reward signal encourages coordination between agents which in turn can benefit from inter-agent communications. 

Furthermore, at each time step $t$ agents choose a communication message to send to the other agent by selecting a communication action $\rc_i(t) \in \mathcal{C} = \{0,1\}^R$ of $R$ bits, where $R$ (bits per channel use / per time step) is the fixed bit-budget of all inter-agent communication channels.
 The goal of the MAS is to maximize the average return by solving the problem (\ref{decentralized problem - State Aggregation}).

\vspace{-5mm}
\subsection{Conventional Information Compression In multi-agent Coordination Tasks} \label{conventional communication - subsection} 
\vspace{-2mm}
As a baseline, we consider a conventional scheme that selects communications and actions separately. For communication, each agent $i$ sends its observation $\ro_i(t)$ to the other agents by following policy $\pi^c_i(\cdot)$. According to this policy the agent's observation $\ro_i(t)$ will be mapped to a binary bit sequence $\rc_i(t)$, using an injective (and not necessarily surjective) mapping $f_1: \Omega \rightarrow \{0,1\}^R$. Consequently, the communication policy $\pi^c_i$ becomes deterministic and follows 
\textcolor{Mycolor3}{
\begin{equation}
    \pi^c_i\big(\rc_i(t+1)| \ro_i(t) \big) = \delta\Big(\! \rc_i(t+1) - f_1\big( \ro_i(t) \big)\!\Big)\!.
\end{equation}}
Agent $i$ obtains an estimate $ {\rc}_j(t)$ of the observation of all agents $j \in \mathcal{N}_{-i}$ by having access to a quantized version of $\ro_j(t)$. This estimate is used to define the environment state-action value function $Q^m_j \big(\ro_i(t), {\rc}_{-i}(t),\rm_i(t)\big)$. This function is updated using Q-learning and the UCB policy in a manner similar to Algorithm 1, with no communication policy to be learned.

This communication strategy is proven to be optimal \cite{pynadath2002communicative}, if the inter-agent communication does not impose any cost on the cooperative objective function, the communication channel is noise-free and the bit-budget of communication channels are larger than the entropy rate of the observation process $R \geq H(\bo_i)$. Under these conditions, and when the dynamics of the environment are deterministic, each agent $i$ can distributively learn the optimal policy $\pi^m_i(\cdot)$, using value iteration or its model-free variants e.g., Q-learning \cite{lauer2000distributedQ}. While this communication policy is optimal only with a channel bit-budget $R \geq H(\bo_j)$, in this paper, we are focused on the scenarios with $R \leq H(\bo_j)$. Therefore, due to the bit-budget of the communication channel, a form of TODC is required.

Note that compression before a converged action policy is not possible, since all observations are a priori equally likely. Thus, we first train the CIC on a communication channel with unlimited capacity. Afterwards, when a probability distribution for observations is obtained, by applying Lloyd's algorithm \cite{lloyd1982least}, we define an equivalence relation on the observation space $\Omega$ with $2^{R}$ numbers of equivalence classes $\mathcal{Q}_1,..., \mathcal{Q}_{2^{R}}$. According to the defined equivalence relation by Lloyd's algorithm, we can uniquely define the mapping $f_1: \Omega \rightarrow \{0,1\}^R$ that maps each agent $i$'s observation $\ro_i(t)$ to a communication message $\rc_i(t)$. The inverse $f^{-1}_1(\cdot)$ of the quantization mapping that maps agent $j$'s quantized observation $\rc_j(t)$ into a estimated observation is not an injective mapping anymore. That is, by receiving the communication message $ {\rc}_j(t) \in \mathcal{Q}_k \subset \mathcal{C}$ agent $i$ can not retrieve $\ro_j(t)$ but understands the observation of agent $j$ has been a member of $\mathcal{Q}_k$. Note that CIC algorithm has a limitation, as it requires the first round of training to be done over communication channels with unlimited capacity.

\vspace{-2mm}
\subsection{Results}

To perform our numerical experiments, rewards of the rendezvous problem are selected as $C_1=1$ and $C_2=10$, while the discount factor is $\gamma = 0.9$. A constant learning rate $\alpha=0.07$ is applied, and the UCB exploration rate $c=12.5$. In any figure that the performance of each scheme is reported in terms of the averaged discounted cumulative rewards, the attained rewards throughout training iterations are smoothed using a moving average filter of memory equal to 10\% of the experiment iterations. \textcolor{Mycolor3}{We will use the terms "value of the collaborative objective function", "value of the objective function" and "average return" interchangeably throughout this section.} Regardless of the grid-world's size and goal location, the grids are numbered row-wise starting from the left-bottom as shown in Fig. \ref{fig: rendezvous problem}-a. Apart from Fig. \ref{fig: C3 3-agent Comparisson} that illustrates the result related to a rendezvous problem for a three-agent system, other figures have been obtained when experimenting in a two-agent environment.
\textcolor{Mycolor3}{
Fig. \ref{fig: C2 Comparisson of all} illustrates the performance of the proposed SAIC as well as six other benchmark schemes
\begin{itemize}
    \item Centralized Q-learning under perfect communications.
    \item Learning based information compression (LBIC) is a different indirect scheme to design task-oriented communications which performs the joint design of communication and control policies through reinforcement learning following an algorithm similar to the one proposed in \cite{mostaani2019Learning}.
    \item CIC, \textcolor{Mycolor3}{see the details of CIC in subsection \ref{conventional communication - subsection}}.
    \item \textcolor{Mycolor3}{ Heuristic non-communicative (HNC) algorithm is a direct heuristic scheme which exploits the domain knowledge of its designer about the rendezvous task - making it not applicable to any other task rather than the rendezvous problem. The domain knowledge is utilized to design a control policy where no communication is present. In HNC, agents approach the goal point and wait next to it for a large enough number of time-steps to make sure the other agent has also arrived there. Only after that, they will get into the goal point. Note that this scheme requires communication/coordination between agents prior to the starting point of the task. }
    \item \textcolor{Mycolor3}{ Heuristic optimal communication (HOC) algorithm is a direct heuristic scheme which exploits the domain knowledge of its designer about the rendezvous task - making it not applicable to any other task rather than the rendezvous problem. The domain knowledge is utilized to design jointly optimal communication and control policies. In HNC, agents approach the goal point and wait next to it until they hear from the other agent it also has arrived there. Only after that, they will get into the goal point. Note that this scheme requires communication/coordination between agents prior to the starting point of the task.} 
    \item Hybrid scheme uses the abstract representation of agents' observations according to SAIC with $R=2$ bits and feeds these latent observations to a centralized controller. The central controller learns the joint action selection of both agents using Q-learning. 
\end{itemize}}

 \textcolor{Mycolor3}{ It is imperative to recall that, not all the schemes evaluated by Fig. \ref{fig: C2 Comparisson of all} are benefit from indirect designs - making them not sufficiently general to be applied to all other multi-agent communication problems with rate-limited inter-agent channels. Regardless of their effectiveness, SAIC, LBIC, CIC and Hybrid are indirect schemes potentially applicable to any other task-oriented compression problem. Whereas, HNC and HOC are tailor-made for the rendezvous problem. In other words, the knowledge that we have about the rendezvous task is already embedded in HNC and HOC to enable the most effective communication/control strategies. HNC and HOC, however, allow us to understand how effective other indirect approaches are even when no knowledge about the specific rendezvous task is embedded in them.} 
 
 The performance is measured in terms of the expected sum of discounted rewards in a rendezvous problem. The grid-world is considered to be of size $N=8$ and its goal location to be $\omega^T=22$. The bit-budget of the channel between the two agents is $R=2$ bits per time step. Since centralized Q-learning is not affected by the limitation on the channel's bit-budget, it achieves optimal performance after sufficient training, 160k iterations. The CIC, due to the insufficient bit-budget of the communication channel, never achieves the optimal solution. The LBIC, however, is seen to outperform the CIC, although it is trained and executed fully distributedly. \textcolor{Mycolor3}{ While enjoying a fast convergence, it is observed that the SAIC can achieve optimal performance by less than 1\% gap, whereas the performance gap for the LBIC and CIC are much more pronounced ranging from 20\% to 30\%. The yellow curve showing the performance of the CIC with no communication between agents would show us the best performance of distributed reinforcement learning that can be achieved if no communication between agents is in place \textcolor{Mycolor3}{without having any domain knowledge - that is present in the HOC and HNC}. \textcolor{Mycolor3}{ In fact, the better performance of any scheme compared with the yellow curve, is the sign that the scheme is either benefiting from some effective communication between agents or from some domain knowledge.} Note that, when inter-agent communication is unavailable, i.e., $R=0$ bit per time step, there would be no difference in the performance of the CIC, SAIC or LBIC as all of them use the same algorithm to find out the action policy $\pi^m_i(\cdot)$. We also recall the fact that both the CIC and SAIC require a separate training phase which is not captured by Fig. 5. SAIC requires a centralized training phase \textcolor{Mycolor3}{ - to perform the computations demonstrated in line 5 of the algorithm 1 -}} and CIC a distributed training phase with unlimited capacity of inter-agent communication channels. The performance of these two algorithms in Fig. 5 is plotted after the first phase of training.

Similar to Fig. \ref{fig: C2 Comparisson of all}, the performance of SAIC is illustrated in Fig. \ref{fig: C3 3-agent Comparisson}, this time in a $n=3$ three-agent system. In this case, the grid-world is considered to be of size $N=3$ and its goal location to be $\omega^T=9$. The bit-budget of the inter-agent communication channels is set to be $R=1$ bits per time step. The shaded area around the curve corresponding to SAIC, shows the standard deviation of SAIC in the training as well as the execution phases - at any given training episode $k$ the width of the shaded curve is equal to the standard deviation of SAIC's return from the training episode $k$ to the episode $k-1000$. This figure illustrates the very robust performance of SAIC in a three-agent scenario. For this particular experiment we used decaying epsilon greedy policies with the starting value of $\epsilon = 1$ and the ending value of $\epsilon = 0.03$. To overcome the issue of credit assignment in multi-agent systems - see e.g., \cite{FoersterCounter} to get familiar with the concept, here we used a different reward function via which we trained the agents. Accordingly, given observations $\langle \ro_i(t+1), ..., \ro_n(t+1) \rangle$ and actions $ \langle \rm_1(t+1), ..., \rm_n(t+1) \rangle $, all agents receive a single team reward
{\small\begin{equation}\label{example environment noise distribution}
   r\big( \ro_1(t), ..., \ro_n(t),\rm_1(t), ..., \rm_n(t) \big)=
   \begin{cases}
    C_2^{n'-1}, & \text{if  $P_3$},\\
    0, & \text{otherwise},\\
   \end{cases}
\end{equation}}
where the proposition $P_3$ is defined as $P_3: T\big(\ro_1(t), ..., \ro_n(t),\rm_1(t), ..., \rm_n(t)\big) \in \mathcal{S}^{T'}_n$. When a subset $\mathcal{N}',\,\, |\mathcal{N}| = n' \leq n$ of agent arrives at the target point $\omega^T$, the episode will be terminated with the reward $C_2^{n'-1}$ being obtained, while the largest reward $C_2^{n-1}$ is attained only when all agents visit the goal point at the same time. Note that this reward signal encourages coordination between agents which in turn can benefit from inter-agent communications. 

 \begin{figure}[htbp!] 
  \centering 
      \includegraphics[width=0.42\textwidth]{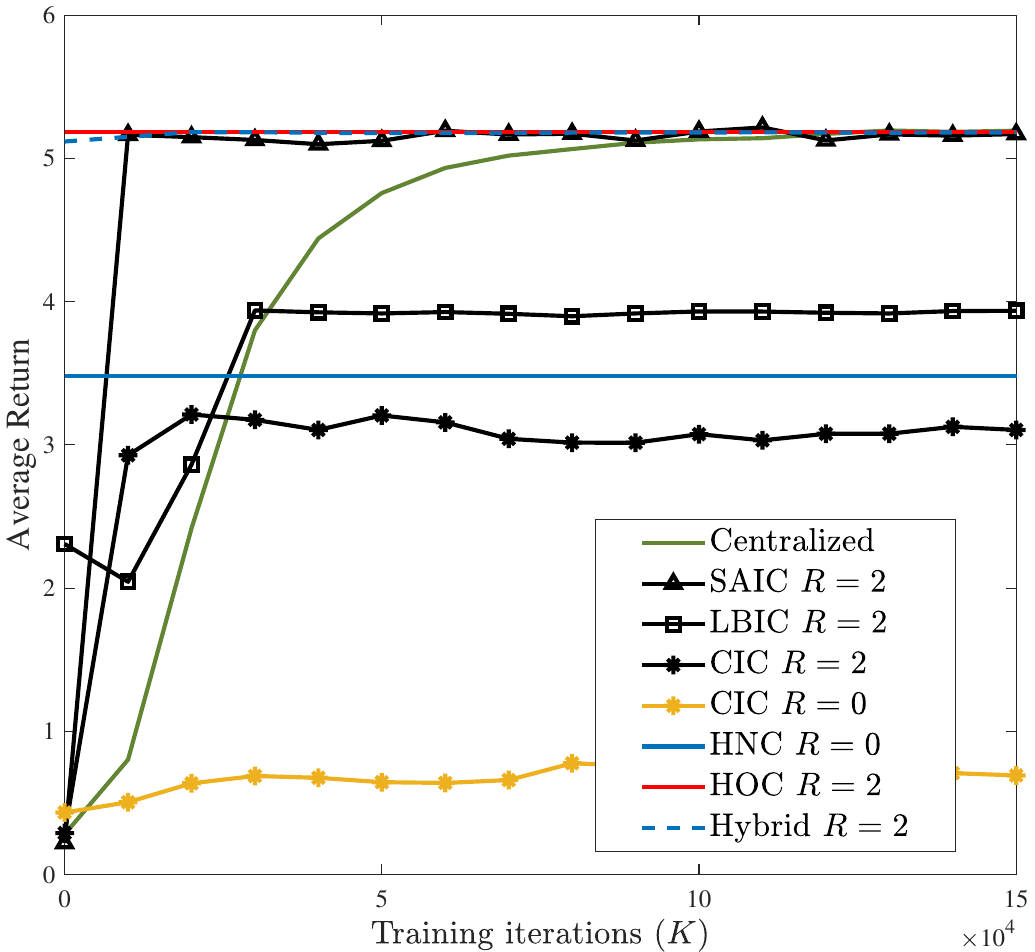} 
      \vspace{-2mm}
  \caption{A comparison between all seven schemes in terms of the achievable objective function with the bit-budget of $R=2$ bits per channel use/time steps and number of training iterations/episodes $K=200k$.}
  \label{fig: C2 Comparisson of all}
\end{figure}

 \begin{figure}[htbp!] 
  \centering 
      \includegraphics[width=0.42\textwidth]{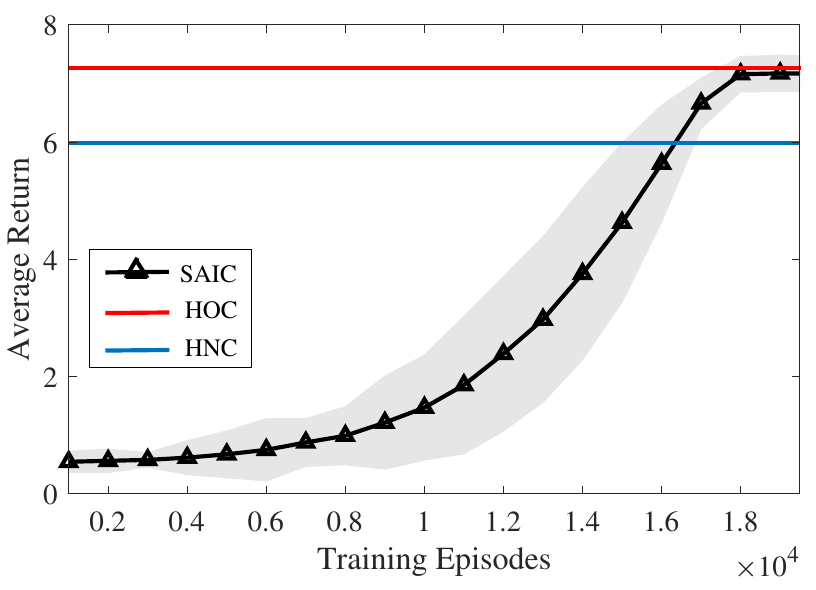} 
      \vspace{-2mm}
  \caption{A comparison between SAIC, HOC and HNC within a three-agent system in terms of the system's average return with the bit-budget of $R=1$ bit per time steps and number of training iterations/episodes $K=20k$. The shaded area around SAIC's curve shows the standard deviation of SAIC in its performance.}
  \label{fig: C3 3-agent Comparisson}
\end{figure}

\textcolor{Mycolor3}{ To explain the underlying reasons for the remarkable performance of the SAIC, Fig. \ref{fig: state-aggregation - centralized learning} is provided so that equivalence classes $\{\mathcal{P}_{i,1}, ..., \mathcal{P}_{i,2^R}\}$ computed by the SAIC can be seen - all the locations of the grid shaded with the same colour belongs to the same $\epsilon$-cost-uniform equivalence class.} The SAIC is extremely efficient in performing state aggregation such that the loss of observation information barely incurs any loss on the achievable sum of discounted rewards - also depicted in Fig. 5. The Fig. \ref{fig: state-aggregation - centralized learning}-(a), illustrates the state aggregation adopted by the SAIC, for which the average return is illustrated in Fig. 4. It is illustrated in Fig. \ref{fig: state-aggregation - centralized learning}-(a) that how the SAIC performs observation compression with ratio $R_c= 3:1$, while it leads to nearly no performance loss for the collaborative task of the MAS. Here the definition of compression ratio follows
$
    R_c = {\ceil*{H\big( \bo_i(t) \big)}}/ {\ceil*{H\big( \bc_i(t) \big)}}.
$
\vspace{-0mm}
 \begin{figure}[htbp!] 
  \centering 
      \includegraphics[width=0.47\textwidth]{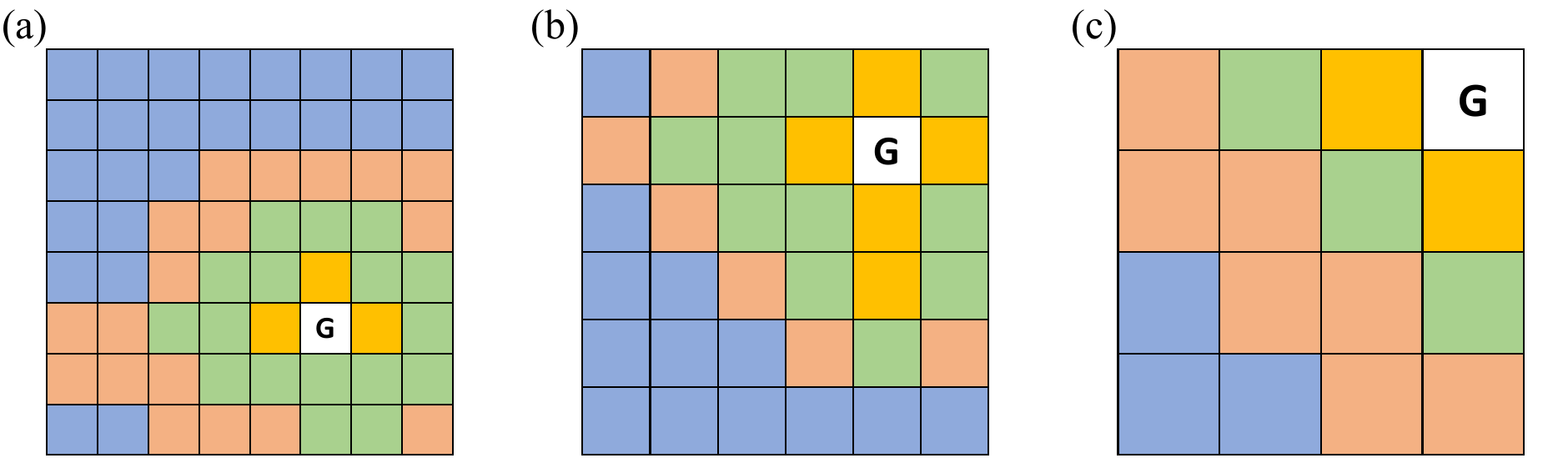} 
      \vspace{-0mm}
  \caption{State aggregation for multi-agent communication in a two-agent rendezvous problem with grid-worlds of varied sizes and goal locations. The observation space is aggregated to four equivalence classes, $R=2$ bits, and the number of training episodes has been $K=1500k$, $K=1000k$ and $K=500k$ for figures (a) and (b) and (c) respectively. Locations with similar colours represent all the agents' observations which are grouped into the same equivalence class. The data compression ratio $R_c$ has been seen to be 6:2, 5:2 and 4:2 in subplots a), b) and c) respectively. It is also observed that the observation clusters identified by SAIC have not been linearly separable under their original representation. In contrast, when clustered according to their values, observation points become linearly separable - see also Fig. \ref{fig: evaluation of value function approximation } .}
  \label{fig: state-aggregation - centralized learning}
\end{figure}
\vspace{-0mm}
It was observed in \ref{fig: state-aggregation - centralized learning} that the observation clusters identified by SAIC have not been linearly separable under their original representation. In contrast, when clustered according to their values,  as seen in Fig. \ref{fig: evaluation of value function approximation },  observation points become linearly separable. Fig. \textcolor{Mycolor3}{ Fig. \ref{fig: evaluation of value function approximation }, allows us to see how precise the approximation of $V_{{\pi^m}^{*},\pi^c}\big(\bo_i(1), {\bc}_{-i}(1)\big)$ by the value function $V^{*}\big(\ro_i(t), {\rc}_{-i}(t)\big)$ is - suggested by lemma \ref{lem: optimal value of the aggregated state}. The figure illustrates the values for both $V_{{\pi^m}^{*},\pi^c}\big(\bo_i(1), {\bc}_{-i}(1)\big)$ and $V^{*}\big(\ro_i(t),\ro_{-i}(t)\big)$, where $\ro_i(t)=21$ and $\ro_{-i}(t)$ can take on possible values in $ \Omega$. For instance the values $7.2$ mentioned on the right down corner of the grid demonstrates the value of $V^{*}\big(\ro_i(t),\ro_j(t)\big)$ when $\ro_i(t)=20$ and $\ro_j(t)=7$. This figure also allows finding the value of $\epsilon$ for all $\epsilon$-cost-uniform groups.}

 \begin{figure}[htbp!] 
  \centering 
      \includegraphics[width=0.49\textwidth]{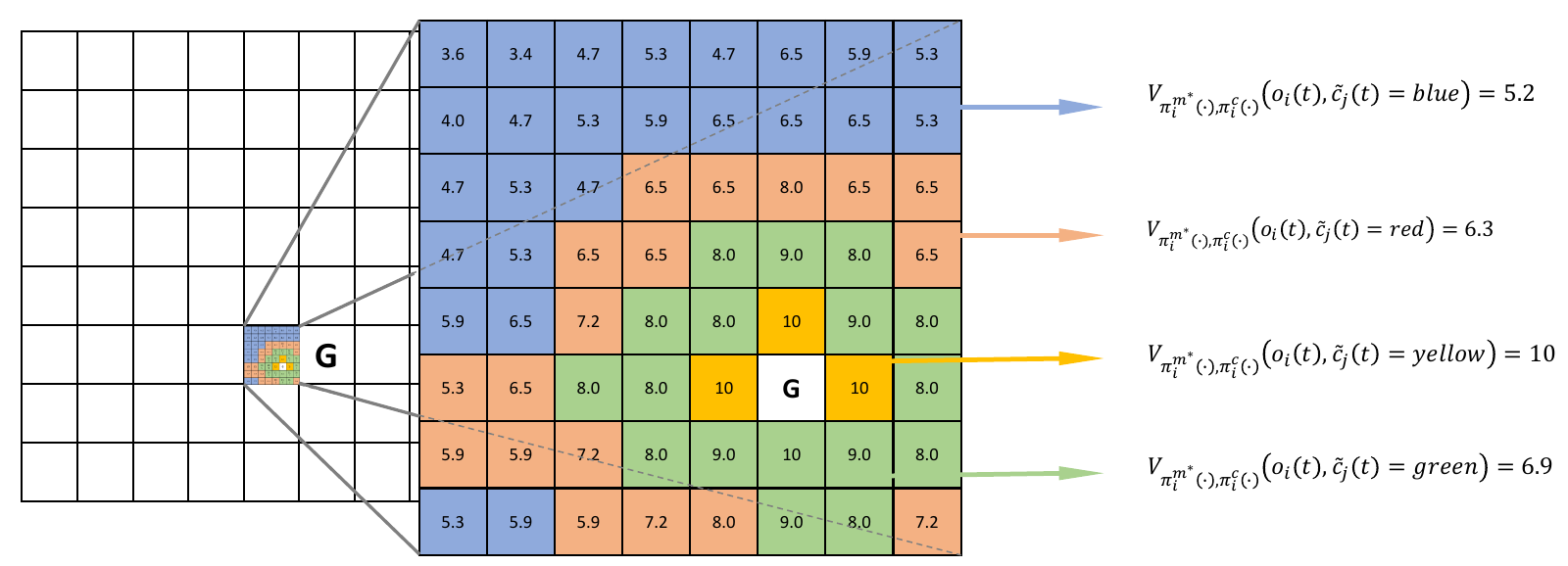}
      \vspace{-3mm}
  \caption{ Left grid-world shows the observation space $\Omega$, amongst which one particular observation is chosen $\ro_i(t)=20$. While agent $i$ makes this observation, agent $j$ can potentially be at any other 64 locations of the greed. The value function $V^{*}\big(\ro_i(t)=20,\ro_j(t) \big)$ for all $\ro_j(t) \in \Omega$ is depicted in the right grid-world, e.g. a number at location 22, shows the value function $V^{*}\big(\ro_i(t)=20,\ro_j(t)=22 \big) = 10$. You can also see the values of $V_{{\pi^m}^{*},\pi^c}\big(\bo_i(t), {\bc}_j(t)\big)$ for $\ro_i(t)=20$ and all possible $ {\rc}_j(t) \in \mathcal{C}$ with $R=2$ bits.} 
  \label{fig: evaluation of value function approximation }
\end{figure}

\textcolor{Mycolor3}{We also investigate the impact of channel bit-budget $R$ on the value of average return achieved by the LBIC, SAIC and CIC, in Fig. \ref{fig: Variable capacities}. In this figure, the normalized value of average return achieved for any scheme at any given \textcolor{Mycolor3}{$R$} is shown.  As per (\ref{normalization}), the average return for the scheme of interest is computed by $\mathbb{E}_{p_{\pi^m,\pi^c}(\{\rtr (t)\}_{t=1}^{t=M})}\big{\{}  \bg(1)  \big{\}}$, where $\pi^m_i(\cdot)$ and $\pi^c_i(\cdot)$ are obtained by the scheme of interest after solving (\ref{decentralized problem - State Aggregation}) with a given value of \textcolor{Mycolor3}{$R$}. The average return is then normalized by dividing it to the average return $ \mathbb{E}_{p_{\pi^{*}}(\{\rtr (t)\}_{t=1}^{t=M})}\big{\{}  \bg(1)  \big{\}} $ that is obtained by the optimal centralized policy $\pi^{*}(\cdot)$. The policy policy $\pi^{*}(\cdot)$ is the optimal solution to (\ref{centralized problem - general problem}) under no communications constraint.}
\begin{equation} \label{normalization}
    \frac{ \mathbb{E}_{p_{\pi^m,\pi^c}(\{\rtr (t)\}_{t=1}^{t=M})}\big{\{}  \bg(1)  \big{\}} }
    { \mathbb{E}_{p_{\pi^{*}}(\{\rtr (t)\}_{t=1}^{t=M})}\big{\{}  \bg(1)  \big{\}} }.
\end{equation}
Accordingly, when the normalized objective function of a particular scheme is seen to be close to the value $1$, it implies that the scheme has been able to compress the observation information with almost zero loss with respect to the achieved objective function. On one hand, it is demonstrated that the SAIC achieves the optimal performance \textcolor{Mycolor3}{ while running with 2 bits of inter-agent communications}, while it takes the CIC at least $\textcolor{Mycolor3}{R}=4$ bits to get to achieve a sub-optimal value of the objective function. The LBIC, on the other hand, provides more than 10\% performance gain in very low rates of communication $R \in \{1,2,3\}$ bits per time step, compared with CIC and 20\% performance gain compared with SAIC at $R=1$ bits per time step.

 \begin{figure}[htbp!] 
  \centering 
      \includegraphics[width=0.40\textwidth]{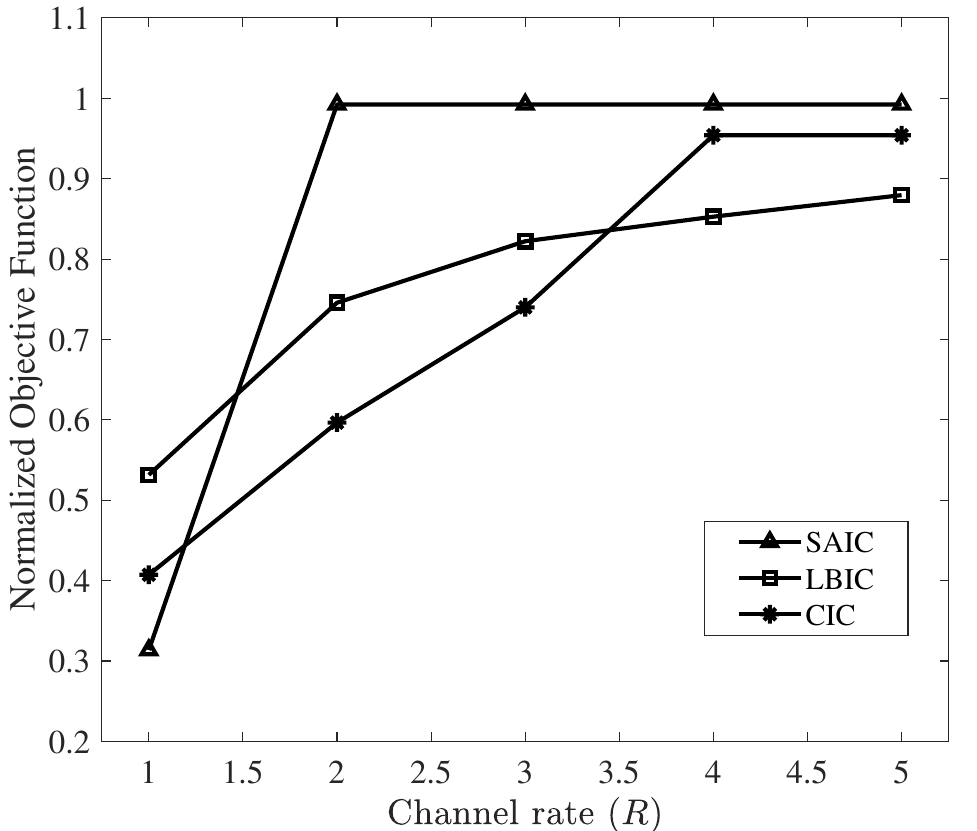} 
    \vspace{-0.2cm}
  \caption{A performance comparison between several multi-agent communication and control schemes under different achievable bit rates. All experiments are performed where $N=8$ and $\omega^T=21$, similar to the grid-world of Fig. \ref{fig: state-aggregation - centralized learning} -a. The number of training episodes/iterations for any scheme at any given channel bit-budget $R$ has been $K=200K$.}
  \label{fig: Variable capacities}
\end{figure}

Fig. \ref{fig: variable compression ratios}, studies the normalized objective functions attained by the LBIC, SAIC and CIC under different compression ratios $R_c$. A whopping 40\% performance gain is acquired by the SAIC, in comparison to the CIC, at high compression ratio $R_c=3:1$. \textcolor{Mycolor3}{This is equivalent to 66\% of saving in the bit-budget with no performance drop with respect to the collaborative objective function.} The SAIC, however, underperforms the LBIC and CIC at very high compression ratio of $R_c=6:1$. This is due to the fact that the condition mentioned in remark 2 is not met at this high rate of compression. Moreover, the CIC scheme is seen not to achieve the optimal performance even at the compression rate of $R_c=6:5$ which is due to the fact that by exceeding the compression ratio $R_c=1:1$ each agent $i$ may lose some information about the observation $\ro_j(t)$ of the other agent which can be helpful in taking the optimal action decision.

 \begin{figure}[htbp!] 
  \centering 
      \includegraphics[width=0.40\textwidth]{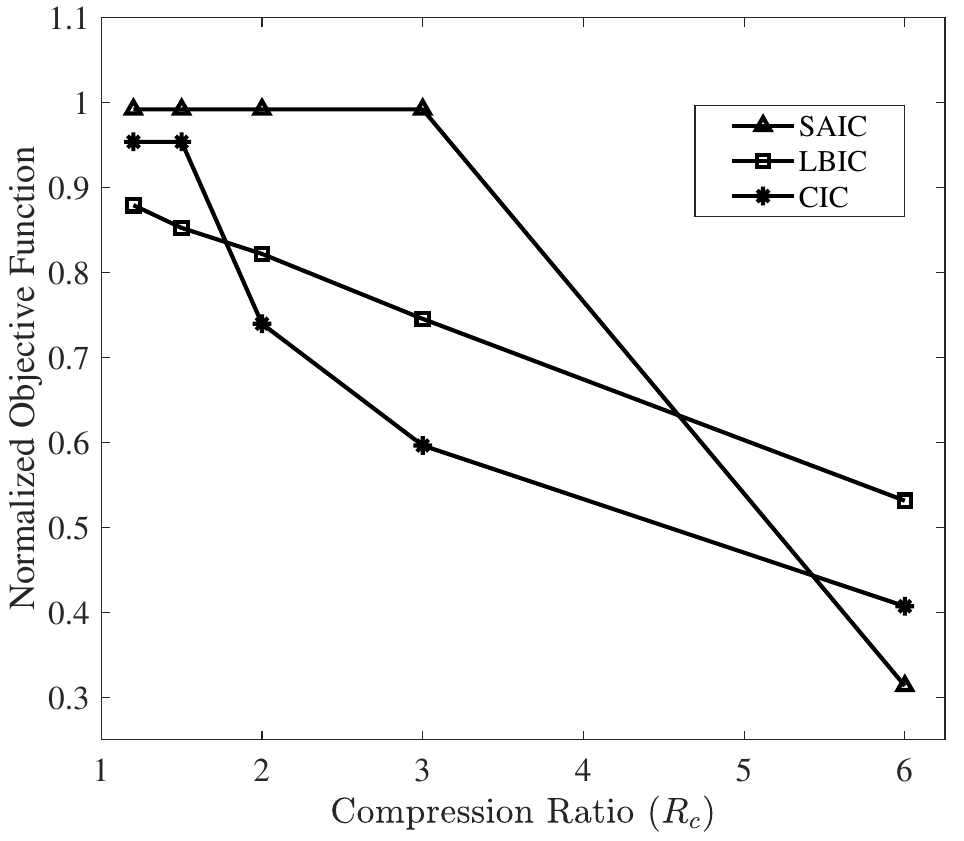} 
      \vspace{-2mm}
  \caption{A performance comparison between several multi-agent communication and control schemes under different rates of data compression. All experiments are performed where $N=8$ and $\omega^T=21$. The number of training episodes/iterations for any scheme at any given bit-budget $R$ has been $K=200K$.}
  \label{fig: variable compression ratios}
  \vspace{-0.2cm}
\end{figure}

 As demonstrated through a range of numerical experiments, the weakness of conventional schemes for compression of agents' observations is that they may lose/keep information regardless of how useful they can be towards achieving the optimal objective function. In contrast, the task-based compression schemes SAIC and LBIC, for communication bit-budgets (very) lower than the entropy of the observation process, manage to compress the observation information not to minimize the distortion but to maximize the achievable value of the objective function. \textcolor{Mycolor3}{Even though the numerical example provided in section IV, evaluates the performance of SAIC in a problem with a very low communication bit-budget, our theoretical results are applicable in scenarios with higher communication rates, as long as the processing unit that is deployed to solve the problem (\ref{centralized problem - general problem}) is of sufficient computational resources to solve the problem in the desire time window.}

\vspace{-0.3cm}
\section{Conclusion} \label{sec: conclusion}
\vspace{-2mm}
We have investigated the distributed joint design of communications and control for an MAS under bit-budgeted communications with the ultimate goal of maximizing the system's expected return. 
Since we consider a limited bit-budget for the multi-agent communication channels, task-based compression of agents' observations has been of the essence. Our proposed scheme, SAIC, which derives and solves the TODC problem can be differentiated from the conventional data quantization algorithms in the sense that it does not aim at achieving minimum possible distortion between the original signal and its reconstructed version - given a bit-budget for inter-agent communications. In contrast, SAIC aims at achieving the minimum possible distortion between the (learned) usefulness/value of the original observation signal and the learned usefulness/value of the the reconstructed observation signal - given a bit-budget for inter-agent communications.
We have demonstrated the outstanding performance of SAIC compared with the conventional data compression algorithms, by up to a remarkable 40\% improvement in the achieved objective function, when being imposed with tight constraints on the communication bit-budget.

To maximize the system's expected return, we could show analytically, how one can disentangle the TODC from the control problem - given the possibility of a centralized training phase. Our analytical studies confirm that despite the separation of the TODC and control problems, we can ensure very little compromise on the MAS's average return - compared with the jointly optimal control and quantization.
\textcolor{Mycolor3}{Since the computational complexity of Q-learning in the centralized training phase is order of $|\Omega^n \times \mathcal{M}^n|$ time complexity \cite{azar2011speedy}, the addition of one single agent will multiply the complexity of the centralized training by $|\Omega \times \mathcal{M}|$. Thus, the complexity of the centralized training phase becomes a hurdle for the scalability of SAIC to a high number of agents. \textcolor{Mycolor3}{ Accordingly, improving the scalability of the algorithm as well as extending the results for non-symmetric variable bit-budgets can be useful avenues to improve the applicability of the proposed schemes.}}

\bibliographystyle{IEEEtran}
{\small
\bibliography{Bibfile}}

\begin{thebibliography}{10}
\providecommand{\url}[1]{#1}
\csname url@samestyle\endcsname
\providecommand{\newblock}{\relax}
\providecommand{\bibinfo}[2]{#2}
\providecommand{\BIBentrySTDinterwordspacing}{\spaceskip=0pt\relax}
\providecommand{\BIBentryALTinterwordstretchfactor}{4}
\providecommand{\BIBentryALTinterwordspacing}{\spaceskip=\fontdimen2\font plus
\BIBentryALTinterwordstretchfactor\fontdimen3\font minus
  \fontdimen4\font\relax}
\providecommand{\BIBforeignlanguage}[2]{{%
\expandafter\ifx\csname l@#1\endcsname\relax
\typeout{** WARNING: IEEEtran.bst: No hyphenation pattern has been}%
\typeout{** loaded for the language `#1'. Using the pattern for}%
\typeout{** the default language instead.}%
\else
\language=\csname l@#1\endcsname
\fi
#2}}
\providecommand{\BIBdecl}{\relax}
\BIBdecl

\bibitem{mostaani2021task}
A.~Mostaani, T.~X. Vu, S.~K. Sharma, Q.~Liao, and S.~Chatzinotas,
  ``Task-oriented communication system design in cyber-physical systems: A
  survey on theory and applications,'' \emph{arXiv preprint arXiv:2102.07166},
  2021.

\bibitem{gunduz2022beyond}
D.~Gunduz, Z.~Qin, I.~E. Aguerri, H.~S. Dhillon, Z.~Yang, A.~Yener,
  K.~Kit~Wong, and C.-B. Chae, ``Beyond transmitting bits: Context, semantics,
  and task-oriented communications,'' \emph{arXiv e-prints}, pp. arXiv--2207,
  2022.

\bibitem{strinati20216g}
E.~C. Strinati and S.~Barbarossa, ``6g networks: Beyond shannon towards
  semantic and goal-oriented communications,'' \emph{Computer Networks}, vol.
  190, p. 107930, 2021.

\bibitem{1056251}
H.~Witsenhausen, ``Indirect rate distortion problems,'' \emph{IEEE Transactions
  on Information Theory}, vol.~26, no.~5, pp. 518--521, 1980.

\bibitem{ioannou1988theory}
P.~Ioannou and J.~Sun, ``Theory and design of robust direct and indirect
  adaptive-control schemes,'' \emph{International Journal of Control}, vol.~47,
  no.~3, pp. 775--813, 1988.

\bibitem{barel2017come}
A.~Barel, R.~Manor, and A.~M. Bruckstein, ``Come together: Multi-agent
  geometric consensus,'' \emph{arXiv preprint arXiv:1902.01455}, 2017.

\bibitem{xie2022task}
H.~Xie, Z.~Qin, X.~Tao, and K.~B. Letaief, ``Task-oriented multi-user semantic
  communications,'' \emph{IEEE Journal on Selected Areas in Communications},
  vol.~40, no.~9, pp. 2584--2597, 2022.

\bibitem{shlezinger2020task}
N.~Shlezinger and Y.~C. Eldar, ``Task-based quantization with application to
  mimo receivers,'' \emph{arXiv preprint arXiv:2002.04290}, 2020.

\bibitem{palattella2018enabling}
M.~R. Palattella and N.~Accettura, ``Enabling internet of everything
  everywhere: Lpwan with satellite backhaul,'' in \emph{2018 Global Information
  Infrastructure and Networking Symposium (GIIS)}.\hskip 1em plus 0.5em minus
  0.4em\relax IEEE, 2018, pp. 1--5.

\bibitem{chaari2019heterogeneous}
L.~Chaari, M.~Fourati, and J.~Rezgui, ``Heterogeneous lorawan \& leo satellites
  networks concepts, architectures and future directions,'' in \emph{2019
  Global Information Infrastructure and Networking Symposium (GIIS)}.\hskip 1em
  plus 0.5em minus 0.4em\relax IEEE, 2019, pp. 1--6.

\bibitem{azari2022evolution}
M.~M. Azari, S.~Solanki, S.~Chatzinotas, O.~Kodheli, H.~Sallouha, A.~Colpaert,
  J.~F.~M. Montoya, S.~Pollin, A.~Haqiqatnejad, A.~Mostaani \emph{et~al.},
  ``Evolution of non-terrestrial networks from 5g to 6g: A survey,'' \emph{IEEE
  Communications Surveys \& Tutorials}, 2022.

\bibitem{nair2003exponential}
G.~N. Nair and R.~J. Evans, ``Exponential stabilisability of finite-dimensional
  linear systems with limited data rates,'' \emph{Automatica}, vol.~39, no.~4,
  pp. 585--593, 2003.

\bibitem{nair2004stabilizability}
------, ``Stabilizability of stochastic linear systems with finite feedback
  data rates,'' \emph{SIAM Journal on Control and Optimization}, vol.~43,
  no.~2, pp. 413--436, 2004.

\bibitem{lauer2000distributedQ}
M.~Lauer and M.~A. Riedmiller, ``An algorithm for distributed reinforcement
  learning in cooperative multi-agent systems,'' in \emph{Proc. Conference on
  Machine Learning}.\hskip 1em plus 0.5em minus 0.4em\relax Morgan Kaufmann
  Publishers Inc., 2000.

\bibitem{kostina2019rate}
V.~Kostina and B.~Hassibi, ``Rate-cost tradeoffs in control,'' \emph{IEEE
  Transactions on Automatic Control}, vol.~64, no.~11, pp. 4525--4540, 2019.

\bibitem{tung2021effective}
T.-Y. Tung, S.~Kobus, J.~R. Pujol, and D.~Gunduz, ``Effective communications: A
  joint learning and communication framework for multi-agent reinforcement
  learning over noisy channels,'' \emph{arXiv preprint arXiv:2101.10369}, 2021.

\bibitem{arimoto1972algorithm}
S.~Arimoto, ``An algorithm for computing the capacity of arbitrary discrete
  memoryless channels,'' \emph{IEEE Transactions on Information Theory},
  vol.~18, no.~1, pp. 14--20, 1972.

\bibitem{shlezinger2021deep}
N.~Shlezinger and Y.~C. Eldar, ``Deep task-based quantization,''
  \emph{Entropy}, vol.~23, no.~1, p. 104, 2021.

\bibitem{pynadath2002communicative}
D.~V. Pynadath and M.~Tambe, ``The communicative multiagent team decision
  problem: Analyzing teamwork theories and models,'' \emph{Journal of
  Artificial Intelligence Research}, vol.~16, pp. 389--423, Jun. 2002.

\bibitem{lee2020optimization}
D.~Lee, N.~He, P.~Kamalaruban, and V.~Cevher, ``Optimization for reinforcement
  learning: From a single agent to cooperative agents,'' \emph{IEEE Signal
  Processing Magazine}, vol.~37, no.~3, pp. 123--135, 2020.

\bibitem{ZhangCoordinating}
C.~Zhang and V.~Lesser, ``Coordinating multi-agent reinforcement learning with
  limited communication,'' in \emph{Conference on Autonomous Agents and
  Multi-agent Systems}, St. Paul, Minnesota, May 2013, pp. 1101--1108.

\bibitem{fischer2004hierarchical}
F.~Fischer, M.~Rovatsos, and G.~Weiss, ``Hierarchical reinforcement learning in
  communication-mediated multiagent coordination,'' in \emph{Proc. IEEE Joint
  Conference on Autonomous Agents and Multiagent Systems, 2004. AAMAS 2004.},
  New York, Jul. 2004, pp. 1334--1335.

\bibitem{kasai2008learning}
T.~Kasai, H.~Tenmoto, and A.~Kamiya, ``Learning of communication codes in
  multi-agent reinforcement learning problem,'' in \emph{Soft Computing in
  Industrial Applications, 2008. SMCia'08. IEEE Conf. on}.\hskip 1em plus 0.5em
  minus 0.4em\relax IEEE, 2008, pp. 1--6.

\bibitem{wu2011online}
F.~Wu, S.~Zilberstein, and X.~Chen, ``Online planning for multi-agent systems
  with bounded communication,'' \emph{Artificial Intelligence}, vol. 175,
  no.~2, pp. 487--511, Feb. 2011.

\bibitem{Amiri2018CEASE}
A.~{Amini}, A.~{Asif}, and A.~{Mohammadi}, ``Cease: A collaborative
  event-triggered average-consensus sampled-data framework with performance
  guarantees for multi-agent systems,'' \emph{IEEE Transactions on Signal
  Processing}, vol.~66, no.~23, pp. 6096--6109, 2018.

\bibitem{FoersterLearning}
J.~Foerster, Y.~Assael, N.~de~Freitas, and S.~Whiteson, ``Learning to
  communicate with deep multi-agent reinforcement learning,'' in \emph{Proc.
  Advances in Neural Information Processing Systems}, Barcelona, 2016.

\bibitem{mostaani2019Learning}
A.~{Mostaani}, O.~{Simeone}, S.~{Chatzinotas}, and B.~{Ottersten},
  ``Learning-based physical layer communications for multiagent
  collaboration,'' in \emph{2019 IEEE Intl. Symp. on Personal, Indoor and
  Mobile Radio Communications}, Sep. 2019.

\bibitem{mostaani2020state}
A.~Mostaani, T.~X. Vu, S.~Chatzinotas, and B.~Ottersten, ``State aggregation
  for multiagent communication over rate-limited channels,'' in \emph{GLOBECOM
  2020-2020 IEEE Global Communications Conference}.\hskip 1em plus 0.5em minus
  0.4em\relax IEEE, 2020, pp. 1--7.

\bibitem{kim2019schedule}
D.~Kim, S.~Moon, D.~Hostallero, W.~J. Kang, T.~Lee, K.~Son, and Y.~Yi,
  ``Learning to schedule communication in multi-agent reinforcement learning,''
  in \emph{Intl. Conf. on Learning Representations}, 2019.

\bibitem{lowe2019pitfalls}
R.~Lowe, J.~Foerster, Y.-L. Boureau, J.~Pineau, and Y.~Dauphin, ``On the
  pitfalls of measuring emergent communication,'' in \emph{Intl. Conf. on
  Autonomous Agents and MultiAgent Systems}, 2019.

\bibitem{Bertsekas1989AdaptiveAg}
D.~P. {Bertsekas} and D.~A. {Castanon}, ``Adaptive aggregation methods for
  infinite horizon dynamic programming,'' \emph{IEEE Transactions on Automatic
  Control}, vol.~34, no.~6, pp. 589--598, June 1989.

\bibitem{bertsekas2018feature}
D.~P. Bertsekas, ``Feature-based aggregation and deep reinforcement learning: A
  survey and some new implementations,'' \emph{IEEE/CAA Journal of Automatica
  Sinica}, vol.~6, no.~1, pp. 1--31, 2018.

\bibitem{abel2016near}
D.~Abel, D.~Hershkowitz, and M.~Littman, ``Near optimal behavior via
  approximate state abstraction,'' in \emph{International Conference on Machine
  Learning}.\hskip 1em plus 0.5em minus 0.4em\relax PMLR, 2016, pp. 2915--2923.

\bibitem{rubino1989weak}
G.~Rubino, ``On weak lumpability in markov chains,'' \emph{Journal of Applied
  Probability}, vol.~26, no.~3, pp. 446--457, 1989.

\bibitem{bertsekas2019biased}
D.~Bertsekas, ``Biased aggregation, rollout, and enhanced policy improvement
  for reinforcement learning,'' \emph{arXiv preprint arXiv:1910.02426}, 2019.

\bibitem{zou2018decision}
H.~Zou, C.~Zhang, S.~Lasaulce, and et~al, ``Decision-oriented communications:
  Application to energy-efficient resource allocation,'' in \emph{Intl. Conf.
  on Wireless Networks and Mobile Communications}.\hskip 1em plus 0.5em minus
  0.4em\relax IEEE, 2018.

\bibitem{mao2019learning}
H.~Mao, Z.~Zhang, Z.~Xiao, Z.~Gong, and Y.~Ni, ``Learning agent communication
  under limited bandwidth by message pruning,'' \emph{arXiv preprint
  arXiv:1912.05304}, 2019.

\bibitem{sukhbaatar2016learning}
S.~Sukhbaatar, R.~Fergus \emph{et~al.}, ``Learning multiagent communication
  with backpropagation,'' in \emph{Proc. Advances in Neural Information
  Processing Systems}, Barcelona, 2016, pp. 2244--2252.

\bibitem{stavrou2022rate}
P.~A. Stavrou and M.~Kountouris, ``A rate distortion approach to goal-oriented
  communication,'' 2022.

\bibitem{oliehoek2008optimal}
F.~A. Oliehoek, M.~T. Spaan, and N.~Vlassis, ``Optimal and approximate q-value
  functions for decentralized pomdps,'' \emph{Journal of Artificial
  Intelligence Research}, vol.~32, pp. 289--353, 2008.

\bibitem{monahan1982state}
G.~E. Monahan, ``State of the art—a survey of partially observable markov
  decision processes: theory, models, and algorithms,'' \emph{Management
  science}, vol.~28, no.~1, pp. 1--16, 1982.

\bibitem{zilber2001communication}
\BIBentryALTinterwordspacing
P.~Xuan, V.~Lesser, and S.~Zilberstein, ``Communication decisions in
  multi-agent cooperation: Model and experiments,'' in \emph{Proceedings of the
  Fifth International Conference on Autonomous Agents}, ser. AGENTS
  ’01.\hskip 1em plus 0.5em minus 0.4em\relax New York, NY, USA: Association
  for Computing Machinery, 2001, p. 616–623. [Online]. Available:
  \url{https://doi.org/10.1145/375735.376469}
\BIBentrySTDinterwordspacing

\bibitem{oliehoek2007dec}
F.~A. Oliehoek, M.~T. Spaan, N.~Vlassis \emph{et~al.}, ``{DEC-PoMDPs} with
  delayed communication,'' in \emph{Proc. Multi-agent Sequential
  Decision-Making in Uncertain Domains}, Honolulu, Hawaii, May 2007.

\bibitem{larrousse2018coordination}
B.~Larrousse, S.~Lasaulce, and M.~R. Bloch, ``Coordination in distributed
  networks via coded actions with application to power control,'' \emph{IEEE
  Trans. on Information Theory}, vol.~64, no.~5, pp. 3633--3654, 2018.

\bibitem{Suttonintroduction}
R.~S. Sutton and A.~G. Barto, \emph{Introduction to reinforcement learning},
  2nd~ed.\hskip 1em plus 0.5em minus 0.4em\relax MIT Press, Nov. 2017, vol.
  135.

\bibitem{rizk2018decision}
Y.~Rizk, M.~Awad, and E.~W. Tunstel, ``Decision making in multiagent systems: A
  survey,'' \emph{IEEE Transactions on Cognitive and Developmental Systems},
  vol.~10, no.~3, pp. 514--529, 2018.

\bibitem{boutilier1999multiagent}
C.~Boutilier, ``Multiagent systems: Challenges and opportunities for
  decision-theoretic planning,'' \emph{AI magazine}, vol.~20, no.~4, pp.
  35--35, 1999.

\bibitem{jaakkola1994convergence}
T.~Jaakkola, M.~I. Jordan, and S.~P. Singh, ``Convergence of stochastic
  iterative dynamic programming algorithms,'' in \emph{Advances in neural
  information processing systems}, 1994, pp. 703--710.

\bibitem{heylighen2016stigmergy}
F.~Heylighen, ``Stigmergy as a universal coordination mechanism i: Definition
  and components,'' \emph{Cognitive Systems Research}, vol.~38, pp. 4--13,
  2016.

\bibitem{oliehoek2016concise}
F.~A. Oliehoek, C.~Amato \emph{et~al.}, \emph{A concise introduction to
  decentralized POMDPs}.\hskip 1em plus 0.5em minus 0.4em\relax Springer, 2016,
  vol.~1.

\bibitem{yuksel2013jointly}
S.~Y{\"u}ksel, ``Jointly optimal lqg quantization and control policies for
  multi-dimensional systems,'' \emph{IEEE Transactions on Automatic Control},
  vol.~59, no.~6, pp. 1612--1617, 2013.

\bibitem{lloyd1982least}
S.~Lloyd, ``Least squares quantization in pcm,'' \emph{IEEE transactions on
  information theory}, vol.~28, no.~2, pp. 129--137, 1982.

\bibitem{linde1980algorithm}
Y.~Linde, A.~Buzo, and R.~Gray, ``An algorithm for vector quantizer design,''
  \emph{IEEE Transactions on communications}, vol.~28, no.~1, pp. 84--95, 1980.

\bibitem{FoersterCounter}
J.~N. Foerster, G.~Farquhar, T.~Afouras, N.~Nardelli, and S.~Whiteson,
  ``Counterfactual multi-agent policy gradients,'' in \emph{Thirty-Second AAAI
  Conference on Artificial Intelligence}, 2018.

\bibitem{amato2009incremental}
C.~Amato, J.~S. Dibangoye, and S.~Zilberstein, ``Incremental policy generation
  for finite-horizon dec-pomdps,'' in \emph{Nineteenth International Conference
  on Automated Planning and Scheduling}, 2009.

\bibitem{azar2011speedy}
M.~G. Azar, R.~Munos, M.~Ghavamzadaeh, and H.~J. Kappen, ``Speedy q-learning,''
  2011.

\end{thebibliography}
\vspace{-4mm}
\appendices
\vspace{-2mm}

\vspace{-0mm}

\section{Proof of Theorem \ref{The main theorem}} \label{appendix: proof of the main theorem}
To prove this theorem we first introduce a definition in subsection \ref{subsect: def: TBIC}, together with two lemmas and their proofs in subsections \ref{subsect: lem: Adam's law on the value function} and \ref{subsect: proof: lem: optimal value of the aggregated state}. Lastly, we complete the proof of Theorem \ref{The main theorem}, in subsection \ref{subsect: proof: The main theorem} leveraging the above-mentioned.

\subsection{Task-based information compression problem: a definition}\label{subsect: def: TBIC}
\textcolor{Mycolor3}{
\begin{definition}\label{definition: TBIC}[Task-based information compression (TBIC) problem] 
Let the higher order function $\Pi^{m^*}$ be a map from the vector space $\mathcal{K}^c$ of all possible joint communication policies $\pi^c = \langle \pi^c_1(\cdot), ..., \pi^c_n(\cdot)\rangle$ to the vector space $\mathcal{K}^m$ of optimal corresponding joint control policies $\pi^m = \langle \pi^{m^*}_1(\cdot), ...,  \pi^{m^*}_n(\cdot) \rangle$. Upon the availability of $\Pi^{m^*}$, by plugging it into the problem (\ref{decentralized problem - State Aggregation}), we will have a new problem 
{\small
\begin{align}\label{decentralized problem - communications}
& \underset{\pi^c_i}{\text{max }} 
& & \mathbb{E}_{p_{\Pi^{m^*}, \pi^c}
\big(\{\rtr (t)\}_{t=1}^{t=M}\big)}                  \Big{\{}
   \bg(1) 
\Big{\}},
  \; \; \; \; i \in \mathcal{N} \; \notag\\
& \text{s.t.} & &  
     \textcolor{Mycolor3}{log_2|\mathcal{C}| \leq R},
\end{align}}
where we maximize the system's return only with respect to the joint communication policies $\pi^c$. The joint optimal control policies $\langle \pi^{m^*}_1(\cdot), ..., \pi^{m^*}_n(\cdot) \rangle$ are automatically computed by the mapping $\Pi^{m^*}\big( \pi^c_1(\cdot), ..., \pi^c_n(\cdot) \big)$. The problem is called here as the TBIC problem.
\end{definition}
}

\subsection{Reformulating the objective function: a lemma} \label{subsect: lem: Adam's law on the value function}

\vspace{-2mm}\begin{lemma}\label{lem: Adam's law on the value function}
  The objective function of the decentralized problem (\ref{decentralized problem - State Aggregation}) can be expressed as
{\small
\begin{align} \label{g(t) and value function - State Aggregation}
  & \mathbb{E}_{p_{\pi^m,\pi^c}(\{\rtr (t)\}_{t=t'}^{t=M})}\big{\{}  \bg(t')  \big{\}} = \notag \\
  & \mathbb{E}_{p_{\pi^m,\pi^c}(h_i(\bs(t')))} \Big{\{}
             \!\mathbb{E}_{p_{\pi^m\!,\pi^c}(\{\rtr (t)\}_{t=2}^{t=M} | h_i(\bs(t')))}\!  \big{\{}\!      \bg(t') |  h_i(\bs(t'))
                         \!\big{\}}
                    \!    \Big{\}}\! =\notag\\
&                                             \mathbb{E}_{p_{\pi^m,\pi^c}( h_i(\bs(t')))} \Big{\{}
                        V_{\pi^m,\pi^c}\big( h_i(\bs(t'))\big)
                                   \Big{\}},
\end{align}}
for all $i \in \mathcal{N}$, where $V_{\pi^m,\pi^c}\big(  h_i(\bs(t')) \big)$
is the solution to the Bellman equation corresponding to the joint control and communication policies $\pi^m,\pi^c$.
\end{lemma}

\begin{proof}
    Considering the definition of the value function, given in (\ref{eq: define value function}), the proof is immediately concluded when applying Adam's law on the expectation of the value function
\textcolor{Mycolor3}{ {\small
\begin{align} \label{eq: define value function}
    & V_{\pi^m,\pi^c}\big(h_i\big( \rs(t') \big)  \big) =
    \mathbb{E}_{p_{\pi^m,\pi^c}(\{\rtr (t)\}_{t=t'+1}^{t=M})}\big{\{}  \bg(t') | h_i\big( \rs(t') \big) \big{\}}.
\end{align}}
}
\vspace{-2mm}
\end{proof}

\subsection{Value of the perceived state of environment: a lemma} \label{subsect: proof: lem: optimal value of the aggregated state}

 \begin{lemma} \label{lem: optimal value of the aggregated state}
\textcolor{Mycolor3}{Using the knowledge of the solution $\pi^*(\cdot)$ to the centralized problem, we can find the optimal value of a perceived state $V^{*}\big( h_i\big(\rs(t)\big) \big)$ in terms of the value of the underlying state $V^{*}\big( \rs(t) \big)$ by
{\small
\begin{equation}\label{lemma 1 - equation}
    V^{*}\!\big(h_i\big(\rs(t)\big)\big) \!\! = \!\!   \sum_{\ro_1(t) \in \Omega} ... \sum_{\ro_{n}(t) \in \Omega}
    \!\!\! V^{*}\big( \rs(t) \big) \,\, p\big(\ro_{-i}(t)| \rc_{-i}(t)\big).
\end{equation}}}
\end{lemma}

\vspace{-0mm}
\begin{proof}
\vspace{-2mm}
\begin{align}
    & V^{*}\big(  h_i\big( \rs(t') \big) )\big) = \\ 
    & \mathbb{E}_{p\big( \{tr\}^{M}_{t^{'}} | h_i\big( \rs(t') \big)  \big)}  \Big{\{}
    {\sum}_{t=t^{'}}^M \gamma^{t-1} r\big(\bs(t),\bm(t)\big)
     \big{|} h_i\big( \rs(t') \big) 
                \Big{\}} = \notag \\
    & \mathbb{E}_{p\big( \{\rtr\}^{M}_{t^{'}} | h_i\big( \rs(t') \big)  \big)}
                \Big{\{}
     \bg(t^{'})
     | h_i\big( \rs(t') \big) 
                \Big{\}} = \\
    &            {\sum}_{\{\rtr\}_{t^{'}}^{M}} \bg(t^{'}) 
    p\Big( \{ \rtr \}_{t^{'}}^{M} | h_i\big( \rs(t') \big)  \Big), \notag
\end{align}
\vspace{-0mm}
where the conditional probability $p\Big( \{ \rtr \}_{t^{'}}^{M} | h_i\big( \rs(t') \big)  \Big)$ can be extended following the law of total probabilities 
\begin{align} \label{eq: total prob for approximated value of the perceived state}
   & V^{*}\big(h_i\big( \rs(t') \big) \big) = 
        {\sum}_{\{\rtr\}_{t^{'}}^{M}} \bg(t^{'}) 
    \Bigg[ \sum_{\ro_1(t) \in \Omega } ... \sum_{\ro_n(t) \in \Omega } \notag \\
    &  p\Big( \{ \rtr \}_{t^{'}}^{M} | \ro_i(t^{'}), \ro_{-i}(t^{'}), {\rc}_{-i}(t^{'}) \Big) p\big( \ro_{-i}(t') |  {\rc}_{-i}(t') \big)
    \Bigg],
\end{align}
\vspace{-0mm}
where $\ro_{-i}(t')$ is the observation vector of all agents $i\in \mathcal{N}_{-i}$. In eq. (\ref{eq: total prob for approximated value of the perceived state}) $\ro_i(t^{'}),\ro_{-i}(t^{'})$ are sufficient statistics and can be replaced by $\rs(t')$ and the second summation can be shifted to have
\begin{align}
    & V^{*}\big(h_i\big( \rs(t') \big)\big) = \notag \\
    & \sum_{\ro_1(t) \in \Omega } \!\! ... \!\! \sum_{\ro_n(t) \in \Omega }\sum_{\{\rtr\}_{t^{'}}^{M}} \!\! \bg(t^{'}) 
     p\Big( \{ \rtr \}_{t^{'}}^{M} | \rs(t')) \Big)
     p\big( \ro_{-i}(t) |  {\rc}_{-i}(t) \big),
\end{align}
\vspace{-0mm}
where $\sum_{\{\rtr\}_{t^{'}}^{M}} \bg(t^{'}) 
     p\Big( \{ \rtr \}_{t^{'}}^{M} | \rs(t')) \Big)$ can be replaced with $V^*\big( \rs(t) \big)$, concluding the proof.
\end{proof}
\vspace{-0mm}

\subsection{Proof of Theorem \ref{The main theorem}} \label{subsect: proof: The main theorem}
   \begin{proof}
     \textcolor{Mycolor3}{    Further to the result of lemma \ref{lem: Adam's law on the value function} and eq. (\ref{g(t) and value function - State Aggregation}), the original problem (\ref{decentralized problem - State Aggregation}) can be expressed by
        {\small\begin{equation} \label{simplified decentralized problem - equation}
            \begin{aligned}
                & \underset{\pi^m_i(\cdot),\pi^c_i(\cdot)}{\text{max }} 
                & & \mathbb{E}_{p_{\pi^m,\pi^c}
                \big( h_i \big( \rs(1) \big) \big)}  \Big{\{}
                        V_{\pi^m,\pi^c}\big( h_i \big( \rs(1) \big) \big)
                                        \Big{\}}, \\
                & \text{s.t.} & &
                \textcolor{Mycolor3}{log_2|\mathcal{C}| \leq R},
            \end{aligned}
        \end{equation}}
        for $i \in \mathcal{N}$.} \textcolor{Mycolor3}{ \textcolor{Mycolor3}{Now by following definition \ref{definition: TBIC} and plugging $\Pi^{m^*}(\cdot)$ into the problem (\ref{simplified decentralized problem - equation}) we obtain the TBIC problem}}
        \vspace{-2mm}
\textcolor{Mycolor3}{{\small\begin{align}\label{decentralized problem 1st version - State Aggregation}
& \underset{\pi^c_i(\cdot)}{\text{max }} 
& &
    \mathbb{E}_{p_{\Pi^{m^*}(\pi^c),\pi^c} %
\big( h_i \big( \rs(1) \big) \big)} \Big{\{}
                        V_{{\Pi^{m^*}}(\pi^c),\pi^c}\big( h_i \big( \bs(1) \big) \big)
                                   \Big{\}}, \notag \\
& \text{s.t.} & &
     \textcolor{Mycolor3}{log_2|\mathcal{C}| \leq R}, \,\,\, i \in \mathcal{N}.
\end{align}}
\textcolor{Mycolor3}{ We continue by following lemma \ref{lem: optimal value of the aggregated state}, to be able to substitute $V_{{\Pi^{m^*}}(\pi^c),\pi^c}\big( h_i \big( \rs(1) \big) \big)$ with its approximator $V^{*}\big( h_i \big( \rs(1) \big) \big)$. This brings us to the approximated TBIC problem}
{\small\begin{align}\label{decentralized problem 2nd version - State Aggregation}
& \underset{\pi^c_i(\cdot)}{\text{max }} 
& &
    \mathbb{E}_{p_{\pi^*,\pi^c}
\big( h_i \big( \rs(1) \big) \big)} \Big{\{}
                        V^{*}\big( h_i \big( \bs(1) \big) \big)
                                   \Big{\}}
    \; \; i \in \mathcal{N} \notag \\
& \text{s.t.} & &
     \textcolor{Mycolor3}{log_2|\mathcal{C}| \leq R}.
\end{align}}
Note that the optimizers of the problem (\ref{decentralized problem 2nd version - State Aggregation}) and (\ref{value function rate-distortion - State Aggregation}) are identical since the additional term $\mathbb{E} \big{\{}V^{*}\big( \bs(t) \big) \big{\}}$ is independent from the communication policy $\pi^c_i(\cdot)$. Furthermore, the problem (\ref{value function rate-distortion - State Aggregation}) is now expressed as a form of data quantization problem with mean absolute difference of the value functions $V^{*}\big(\bs(t)\big)$ and $V^{*}\big( h_i \big( \bs(1) \big) \big)$ as the measure of distortion. This interpretation of problem (\ref{value function rate-distortion - State Aggregation}) can be better understood later by seeing the eq. (\ref{eq: generalized data quantization problem}).} 
\textcolor{Mycolor3}{{\small\begin{align}\label{value function rate-distortion - State Aggregation}
&  \underset{\pi^c_i(\cdot)}{\text{min}}
& & \!\!\mathbb{E}_{p_{\pi^m,\pi^c}
\big( h_i\big( \rs(1) \big) \big)}
\Big{\{}\!
V^{*}\big( \bs(1) \big) - V^{*}\big( h_i\big( \bs(1) \big) \big)          \!\!  \Big{\}} \notag \\
& \text{s.t.} 
& &  \!\! \textcolor{Mycolor3}{log_2|\mathcal{C}| \leq R},
\end{align}}}
and since $V^{*}\big( \bs(1) \big)$ is always larger than $V^{*}\big( h_i\big( \bs(1) \big) \big) $, the problem above can also be written as

\textcolor{Mycolor3}{{\small\begin{align}\label{eq: generalized data quantization problem}
&  \underset{\pi^c_i(\cdot)}{\text{min}}
& & \!\!\mathbb{E}_{p_{\pi^m,\pi^c}
\big( h_i\big( \rs(1) \big) \big)}
\Big{\{}
 \big{|} V^{*}\big( \bs(1) \big) - V^{*}\big( h_i\big( \bs(1) \big) \big)  \big{|}          \Big{\}} \notag \\
& \text{s.t.} 
& &  \!\! \textcolor{Mycolor3}{log_2|\mathcal{C}| \leq R},
\end{align}}}
concluding the proof of Theorem \ref{The main theorem}.
\end{proof}

\section{Proof of Lemma \ref{lemma: quantization to k-median}}\label{appendix: proof: lemma: quantization to k-median} 
\begin{proof}
\textcolor{Mycolor3}{The term {\small
$\mathbb{E}_{p_{\pi^m,\pi^c}
\big( h_i\big( \rs(1) \big) \big)}
\Big{\{}\!
V^{*}\big( \bs(1) \big) - V^{*}\big( h_i\big( \bs(1) \big) \big)          \!\!  \Big{\}}$}
can be estimated by computing it over the empirical distribution of $\bs(1)$. Note that the empirical joint distribution of $h_i \big( \bs(1) \big)$ can be obtained by following the communication policy $\pi^c_i(\cdot)$ on the empirical distribution of $\bs(1)$. Therefore, the problem (\ref{value function rate-distortion - State Aggregation}) can be rewritten as
{\small\begin{align}\label{Emperical expectation - State Aggregation}
&  \underset{\pi^c_i(\cdot)}{\text{min}}
& & \sum _{\ro_i(1) \in \Omega} ... \sum_{\ro_n(1) \in \Omega} \Big{|} V^{*}\big(\rs(t)\big) - V^{*}\big( h_i\big( \rs(t)  \big)   \big)\Big{|}, \,\, \forall i \in \mathcal{N} \notag \\                                                                 
& \text{s.t.} & &  \textcolor{Mycolor3}{log_2|\mathcal{C}| \leq R}.
\end{align}}}
\textcolor{Mycolor3}{Quantization levels are disjoint sets $\mathcal{P}_{i,k} \subset \Omega$, where their union $\cup_{k=1}^{2^R}\mathcal{P}_{i,k}$ will cover the entire $\Omega$. Each quantization level is represented by only one communication message $\rc_j(t)=\rc_k \in \mathcal{C}$. Further to lemma \ref{lem: optimal value of the aggregated state}, the value of $V^{*}\big( h_i\big( \bs(t)  \big) \big)$ can be computed by empirical mean (\ref{lemma 1 - equation}).}

\textcolor{Mycolor3}{The quantization problem (\ref{Emperical expectation - State Aggregation}) becomes a k-median clustering problem
\vspace{-2mm}
{\small\begin{equation}\label{k-median - State Aggregation}
\begin{aligned}
&  \underset{\mathcal{P}_i}{\text{min}}
& & \underset{j \in \mathcal{N}_{-i}}{\sum_{\ro_j(t) \in \Omega}} \;\;\sum_{k=1}^{2^R} \;\; \sum_{\ro_i(t) \in \mathcal{P}_{i,k}} \Big{|} V^{*}\big(\ro_i(t), \ro_j(t)\big) - \mu_k \Big{|},     
\end{aligned}
\end{equation}}
where $\mathcal{P}_i = \{\mathcal{P}_{i,1}, ..., \mathcal{P}_{i,2^R}\} $ is a partition of $\Omega$, and the first summation $ \underset{j \in \mathcal{N}_{-i}}{\sum_{\ro_j(t) \in \Omega}}$ is a concatenation of $n-1$ summations each one acting over $\ro_j(t) \in \Omega$ where $j \in \mathcal{N}_{-i}$.}

\textcolor{Mycolor3}{By taking the mean of $V^{*}\big( \rs(t) \big) $ over the empirical distribution of $\ro_j(t), \,\, \forall j \in \mathcal{N}_{i}$, we can also marginalize out $\ro_j(t), \,\, \forall j \in \mathcal{N}_{i}$. Again, it does not change the solution of the problem and we will have
{\small\begin{equation}\label{k-median main - State Aggregation}
\begin{aligned}
&  \underset{\mathcal{P}_i}{\text{min}}
& & {\sum}_{k=1}^{2^R} {\sum}_{\ro_i(t) \in \mathcal{P}_{i,k}} \Big{|} V^{*}\big(\ro_i(t)\big) - \mu^{'}_k \Big{|},  
\end{aligned}
\end{equation}}
in which $\mu^{'}_k = \sum_{\ro_j(t) \in \mathcal{P}_{i,k}} \mu_k$ will approximate $V^{*}\big( \rc_i(t) \big)$.}
\end{proof}

\textcolor{Mycolor3}{ To gain more insight about the meaning of this task-based information compression, it is useful to take a look at the conventional quantization problem which is adapted to our problem setting in eq. (\ref{conventional State Aggregation}), where ${\bc}_j = \pi^c_j\big( \bo_j(1)\big)$. In fact, the compression scheme applied in the CIC, explained in subsection (\ref{conventional communication - subsection}), is obtained by solving the following problem}
\textcolor{Mycolor3}{{\small\begin{equation}\label{conventional State Aggregation}
\begin{aligned}
  \underset{\pi^c_i(\cdot)}{\text{min}}
 {\sum}_{\ro_i(1) \in \Omega} \Big{|}  \ro_i(t) -   {\rc}_i(t)\Big{|}^2, ~~    \text{s.t.}   \,\,\, \textcolor{Mycolor3}{log_2|\mathcal{C}| \leq R},\\
\end{aligned}
\end{equation}}}
which can be solved optimally by the Lloyd's algorithm \cite{lloyd1982least}.

\section{Proof of Lemma \ref{lemma: compute optimal value}} \label{appendix: proof: lemma: compute optimal value}
\vspace{-2mm}
\begin{proof}

\textcolor{Mycolor3}{Further to the law of iterated expectations, $V^{*}\big( \ro_i(t')\big)$ can be expressed as}
\textcolor{Mycolor3}{{\small\begin{align} \label{value function iterated expectation - State aggregation}
   & V^{*}\big(\ro_i(t')\big) = \mathbb{E}_{p(\bo_{-i}(t'))} \Bigg{\{}  \mathbb{E}_{p_{\pi^{*}}(\{\rtr (t)\}_{t=t'+1}^{t=M} | \ro_i(t'), \bo_{-i}(t'))} \bigg{\{} \notag \\
   & \bg(t')
               | \bo_i(t') = \ro_i(t') , \bo_{-i}(t')
               \bigg{\}}              \Bigg{\}} = \\
   & \sum_{\ro_{-i}(t') \in \Omega^{n-1}} \!\!\!   p(\bo_{-i}(t) = \ro_{-i}(t'))      
     \mathbb{E}_{\pi^*} \bigg{\{}
    \bg(t')
               | \ro_i(t') , \ro_{-i}(t')               \bigg{\}} \notag
\end{align}}}
where the expectation of the last term is the optimal value of the state $\rs(t') = \langle \ro_i(t') , \ro_{-i}(t')   \rangle$ of the underlying MDP
\begin{align}
    V^{*}\big( \rs(t') \big) =   \mathbb{E}_{\pi^*} \bigg{\{}
    \bg(t')
               | \ro_i(t') , \ro_{-i}(t')               \bigg{\}}.
\end{align}
Following Bellman optimality equation $V^{*}\big(\rs(t')\big)$ can be obtained by centralized $Q$-learning following
\textcolor{Mycolor3}{{\small
\begin{align} \label{value of o1 and o2 - State aggregation}
        & V^{*}\!\big(\rs(t')\big) \!= \! \underset{\rm \in \mathcal{M}^n}{\text{max }} \; Q^{*}\!\big(\rs(t'),\rm(t')\big)  \\
        &   =  \mathbb{E}_{p_{\pi^{*}}(\{\rtr (t)\}_{t=t'+1}^{t=M} | \ro_i(t'),\ro_{-i}(t'))} \!\Bigg{\{}\! 
    \bg(t')
               | \ro_i(t'), \ro_{-i}(t')
              \!\! \Bigg{\}}. \notag 
\end{align}
}}
\textcolor{Mycolor3}{Using (\ref{value function iterated expectation - State aggregation}) and (\ref{value of o1 and o2 - State aggregation}) we can simply compute $V^{*}\big(\ro_i(t')\big)$ by 
{\small\begin{align} \label{value function iterated expectation simplified - State aggregation}
V^{*}\big(\ro_i(t)\big) = 
   \sum_{\ro_{-i}(t) \in \Omega^{n-1}}    \underset{\rm}{\text{max }} \; Q^{*}\big(\rs(t),\rm(t)\big)   p\big(\bo_{-i}(t) = \ro_{-i}(t)\big).  
\end{align}}}
\end{proof}

\section{Proof of Theorem \ref{theorem: error bound}} \label{appendix: proof: theorem: error bound}
\vspace{-2mm}
\begin{proof}
    Without loss of generality, we have written the proof of this theorem for a two agent scenario to improve the readability. Given the proof for the two-agent system, the extension to a multi-agent system is straightforward. According to the \cite{abel2016near}(Lemma 1), optimal state values of the aggregated MDPs (the environment as is seen by one agent during the decentralized training phase of SAIC) are in a small neighbourhood of the optimal values corresponding to the optimal solution to the original underlying MDP:
\begin{align}
    & \forall \ro_j \in \Omega \text{ and } \text{ and } \forall i \in \{1,2\}, j \neq i:  \notag  \\
    & | V^*(\ro_i, \ro_j ) - V_i^{m}(\ro_i,  {c}^{(k)}_j) | < \frac{2 \epsilon}{(1-\gamma)^2},
    \vspace{-2mm}
\end{align}
where $V_i^m(\cdot)$ is the value function corresponding to $\pi^{m,SAIC}_i(\cdot)$. The communication signal $c_j^{(k)} \in \mathcal{C}$ is agent $j$'s communicated message and at the same time is the $k$-th element of the communication space $\mathcal{C} = \{ c^{(1)}, c^{(2)}, ..., c^{|\mathcal{C}|}\}$ i.e., $c_j^{(k)} = c^{(k)}$. Following the eq. (\ref{g(t) and value function - State Aggregation}), one can write the expected return of the system under centralized scheme as :
\begin{align} \label{eq: expected return original MDP}
        & \mathbb{E}_{p_{\pi^*}(\{\rtr (t)\}_{t=t_0}^{t=M})}\big{\{}  \bg(t_0)  \big{\}} = \mathbb{E} \Big{\{} V^*\big( \bo_i(t_0), \bo_j(t_0) \big) \Big{\}} = \notag \\
        & \sum_{\bo_j\in \Omega} \sum_{\bo_i\in \Omega} V^*\big( \bo_i(t_0), \bo_j(t_0) \big) p_{\bo_i, \bo_j}(\ro_i(t_0), \ro_j(t_0)),
    \end{align}
    where the second expectation is taken over the joint probability distribution $p_{\pi^*}(\ro_i(t_0), \ro_j(t_0))$ of $\bo_i$ and $\bo_j$ when following the action policy $\pi^*(\cdot)$. This equation can be extended for multi-agent case only by taking a summation over each agent's observation space on the left-hand side. Similarly, following the eq. (\ref{g(t) and value function - State Aggregation}), one can write the expected return of the system that is run by SAIC as:
\begin{align} \label{eq: expected return aggregated MDP}
        & \mathbb{E}_{p_{\pi^m,\pi^c}(\{\rtr (t)\}_{t=t_0}^{t=M})}\big{\{}  \bg(t_0)  \big{\}} = \mathbb{E} \Big{\{} V^m\big( \bo_i(t_0),  {\bc}_j^{(k)}(t_0) \big) \Big{\}} = \notag \\
        & \sum_{ k = 1}^{|\mathcal{C}|} \sum_{\bo_i\in \Omega} V^m\big( \bo_i(t_0),  {\bc}_j^{(k)}(t_0) \big) p_{\bo_i,  {\bc}_j}(\ro_i(t_0),  {\rc}_j^{(k)}(t_0)).
    \end{align}
    We can rewrite the joint probability $p_{\bo_i,  {\bc}_j}(\ro_i(t_0),  {\rc}_j^{(k)}(t_0))$ as 
    \begin{align}\label{eq: joint observation and comms prob}
             & p_{\bo_i,  {\bc}_j}(\ro_i(t_0),  {\rc}_j^{(k)}(t_0)) =  \sum_{\ro_j(t_0) \in \mathcal{P}_{i,k}} p_{\bo_i, \bo_j}(\ro_i(t_0), \ro_j(t_0)),
    \end{align}
    where the subset $\mathcal{P}_{i,k} \subset \Omega$ stands for the set of all observation realizations $\ro_j$ that are represented by $ {\rc}^{(k)}_j(t_0)$ according to the policy $\pi^{c,SAIC}_i(\cdot)$.
    Given eq. (\ref{eq: joint observation and comms prob}), one can express eq. (\ref{eq: expected return aggregated MDP}) - the expected return of the MAS under SAIC - also as
    \begin{align} \label{eq: expected return aggregated MDP-modified}
        & \mathbb{E}_{p_{\pi^m,\pi^c}(\{\rtr (t)\}_{t=t_0}^{t=M})}\big{\{}  \bg(t_0)  \big{\}} = \mathbb{E} \Big{\{} V^m\big( \bo_i(t_0),  {\bc}_j^{(k)}(t_0) \big) \Big{\}} = \notag \\
        & \sum_{ k = 1}^{|\mathcal{C}|} \sum_{\ro_j(t_0) \in \mathcal{P}_{i,k}} \sum_{\bo_i\in \Omega} \! V^m\big(\! \bo_i(t_0),  {\bc}_j^{(k)}(t_0) \!\big) p_{\bo_i, \bo_j}(\!\ro_i(t_0), \ro_j(t_0)\!).
    \end{align}
    In order for eq. (\ref{eq: expected return original MDP}) to have the arrangement of its summations similar to eq. (\ref{eq: expected return aggregated MDP-modified}), it is sufficient to break its left-hand summation to two parts
    \begin{align} \label{eq: expected return original MDP-modified}
        & \mathbb{E}_{p_{\pi^*}(\{\rtr (t)\}_{t=t_0}^{t=M})}\big{\{}  \bg(t_0)  \big{\}} = \mathbb{E} \Big{\{} V^*\big( \bo_i(t_0), \bo_j(t_0) \big) \Big{\}} = \notag \\
        & \sum_{ k = 1}^{|\mathcal{C}|} \sum_{\ro_j(t_0) \in \mathcal{P}_{i,k}} \sum_{\bo_i\in \Omega} V^*\big( \bo_i(t_0), \bo_j(t_0) \big) p_{\bo_i, \bo_j}(\ro_i(t_0), \ro_j(t_0)),
    \end{align}
    
    Further to equations (\ref{eq: expected return original MDP-modified})-(\ref{eq: expected return aggregated MDP-modified}), the difference between the achievable expected return of the centralized scheme and SAIC can be explained by 
    \begin{align}
        & \mathbb{E}_{p_{\pi^*}(\{\rtr (t)\}_{t=t_0}^{t=M})}\big{\{}  \bg(t_0)  \big{\}} - \mathbb{E}_{p_{\pi^m_i,\pi^c_i}(\{\rtr (t)\}_{t=t_0}^{t=M})}\big{\{}  \bg(t_0)  \big{\}} =  \notag \\
        & \sum_{ k = 1}^{|\mathcal{C}|} \sum_{\ro_j(t_0) \in \mathcal{P}_{i,k}} \sum_{\bo_i\in \Omega}\! V^*\big(\! \bo_i(t_0), \bo_j(t_0) \! \big) p_{\bo_i, \bo_j}(\!\ro_i(t_0), \ro_j(t_0)\!) \, - \notag \\
        & \sum_{ k = 1}^{|\mathcal{C}|} \sum_{\ro_j(t_0) \in \mathcal{P}_{i,k}} \sum_{\bo_i\in \Omega} \! V^m\big(\! \bo_i(t_0),  {\bc}_j^{(k)}(t_0) \!\big) p_{\bo_i, \bo_j}(\!\ro_i(t_0), \ro_j(t_0)\!).
    \end{align}
    
    We now proceed by factorizing the joint probability $p_{\bo_i, \bo_j}(\ro_i(t_0), \ro_j(t_0))$ which yields
    \small{
        \begin{align}
                  & \mathbb{E}_{p_{\pi^*}(\{\rtr (t)\}_{t=t_0}^{t=M})}\big{\{}  \bg(t_0)  \big{\}} - \mathbb{E}_{p_{\pi^m_i,\pi^c_i}(\{\rtr (t)\}_{t=t_0}^{t=M})}\big{\{}  \bg(t_0)  \big{\}} = \notag \\ 
                  & \sum_{ k = 1}^{|\mathcal{C}|} \sum_{\ro_j(t_0) \in \mathcal{P}_{i,k}} \sum_{\bo_i\in \Omega}\! p_{\bo_i, \bo_j}(\!\ro_i(t_0), \ro_j(t_0)\!) \big[ V^*\big(\! \bo_i(t_0), \bo_j(t_0) \! \big) \, \notag \\
                  & \;\;\;\;\;\;\;\;\;\; \;\;\;\;\;\;\;\;\;\;
                    \;\;\;\;\;\;\;\;\;\; \;\;\;\;\;\;\;\;\;\;
                    \;\;\;\;\;\;\;\;\;\; \;\;
                  - V^m\big(\! \bo_i(t_0),  {\bc}_j^{(k)}(t_0) \!\big) \big] 
        \end{align}
        }
    Since $\big{[}V^*\big( \bo_i(t_0), \bo_j(t_0) \big) -  V^m\big( \bo_i(t_0),  {\bc}_j(t_0) \big)\big{]} $ is upper-bounded by a constant term $\frac{2 \epsilon}{(1-\gamma)^2}$, its weighted sum is also upper bounded by the same term $\frac{2 \epsilon}{(1-\gamma)^2}$. Thus we conclude the proof of Theorem \ref{theorem: error bound}. \textcolor{Mycolor3}{We are unsure if the suggested bound is tight. The results obtained in the performance evaluation indicates a large difference between the bound offered above and the performance bound between SAIC and the optimal centralized control. }
\end{proof}

\end{document}